\documentclass[aps,pra,twocolumn,nosuperscriptaddress]{revtex4-2}
\usepackage{amssymb,amsmath}
\usepackage{bm,bbm}
\usepackage{graphicx}
\usepackage[bookmarks=false]{hyperref}
\hypersetup{colorlinks=true,citecolor=blue,linkcolor=red,%
urlcolor=blue,pdfstartview=FitH,bookmarksopen=true}
\usepackage[T1]{fontenc}
\usepackage{comment}
\usepackage{mathtools}
\usepackage{enumitem}
\usepackage[usenames,dvipsnames]{xcolor}
\usepackage{tikz}
\usetikzlibrary{positioning}

\usepackage{extarrows}

\DeclareMathOperator{\sinc}{sinc}
\newcommand{\ket}[1]{\mbox{$ | #1 \rangle $}}
\newcommand{\bra}[1]{\mbox{$ \langle #1 | $}}

\newcommand{\ketbra}[2]{\mbox{$ | #1 \rangle \langle #2 | $}}
\newcommand{\tr}{\mathrm{tr}}

\newcommand{\ra}{\rangle}
\newcommand{\la}{\langle}
\newcommand{\da}{\dagger}

\newcommand{\cO}{\mathcal{O}}

\newcommand{\cC}{\mathcal{C}}

\newcommand{\cL}{\mathcal{L}}

\newcommand{\oM}{\hat{M}}
\newcommand{\oR}{\hat{R}}
\newcommand{\oS}{\hat{S}}
\newcommand{\oP}{\hat{P}}
\newcommand{\oU}{\hat{U}}

\newcommand{\oH}{\hat{H}}
\newcommand{\oh}{\hat{h}}
\newcommand{\os}{\hat{\sigma}}
\newcommand{\oJ}{\hat{J}}

\newcommand{\E}{\mathrm{e}}
\newcommand{\I}{\mathrm{i}}

\definecolor{light-gray}{gray}{0.95}

\usepackage{amsthm}
\newtheoremstyle{note}      
  {\topsep/2}              	
  {\topsep/2}            	
  {}                        
  {\parindent}             	
  {\itshape}                
  {.---}                    
  {0pt}                     
  {\thmname{#1}\thmnumber{ \itshape#2}\thmnote{ (#3)}} 

\newtheorem{theorem}{Theorem}

\newtheorem{proposition}[theorem]{Proposition}

\theoremstyle{definition}

\theoremstyle{remark}

\begin{document}
%
\title{Robust entanglement buffers based on SWAP interactions}
\author{Ye-Chao Liu}
\email{Corresponding author: liu@zib.de}
\affiliation{Zuse-Institut Berlin, Takustra{\ss}e 7, 14195 Berlin, Germany}
\affiliation{Naturwissenschaftlich-Technische Fakult{\"a}t, Universit{\"a}t Siegen, Walter-Flex-Stra{\ss}e 3, 57068 Siegen, Germany}

\author{Otfried G\"uhne}
\email{otfried.guehne@uni-siegen.de}
\affiliation{Naturwissenschaftlich-Technische Fakult{\"a}t, Universit{\"a}t Siegen, Walter-Flex-Stra{\ss}e 3, 57068 Siegen, Germany}

\author{Stefan Nimmrichter}
\email{stefan.nimmrichter@uni-siegen.de}
\affiliation{Naturwissenschaftlich-Technische Fakult{\"a}t, Universit{\"a}t Siegen, Walter-Flex-Stra{\ss}e 3, 57068 Siegen, Germany}

\date{June 19, 2025}
%

\begin{abstract}

Quantum entanglement is the essential resource for quantum communication and distributed information processing in a quantum network. However, the remote generation over a network suffers from inevitable transmission loss and other technical difficulties.
This paper introduces the concept of entanglement buffers as a potential primitive for preparing long-distance entanglement.
We investigate the filling of entanglement buffers with either one Bell state or a stream of Bell states via SWAP interactions. 
We illustrate their resilience to imperfect interactions, noise, and losses, making the buffers suitable for a realistic quantum network scenario.
Additionally, larger entanglement buffers can always enhance these benefits.
\end{abstract}

\maketitle

Quantum networks are the platform to perform quantum communication and information processing tasks, such as quantum teleportation \cite{Bennett.etal1993, Andersen.etal2013}, quantum key distribution \cite{BB84, Gisin_quantum_2002, Panayi.etal2014}, and quantum clock synchronization \cite{Jozsa.etal2000, Ilo-Okeke.etal2018}. 
They also facilitate the realization of quantum computation, metrology, and sensing in a long-distance and distributed manner \cite{Cirac.etal1999, Komar.etal2014, Degen.etal2017, Komar.etal2014, Proctor.etal2018}, ultimately paving the way towards a quantum internet \cite{Kimble2008, Simon2017, Wehner.etal2018}.

Quantum entanglement \cite{Horodecki.etal2009} is arguably the essential resource in the aforementioned quantum network applications, and its generation and storage have been investigated for many years. Early implementations were based on the light-matter interaction 
\cite{duan_long-distance_2001, 
matsukevich_entanglement_2006, 
chou_measurement-induced_2005, 
eisaman_electromagnetically_2005, 
yuan_experimental_2008}, 
which requires photon detection for post-selection. 
The light-matter interface is naturally suitable for entanglement generation and distribution in quantum networks, but it suffers a serious practical problem: the transmission loss \cite{Gisin_quantum_2002, Lasers2010, Panayi.etal2014, FSOC2017}, which limits the distance of directly linked remote network nodes to under $100$ km \cite{Wengerowsky.etal2019}. A way to overcome the problem is quantum repeaters 
\cite{Briegel.etal1998, 
Dur.etal1999, 
yuan_experimental_2008, 
Sangouard.etal2011, 
bernardes_rate_2011, 
sheng_hybrid_2013, 
behera_demonstration_2019, 
pouryousef_quantum_2022, 
elsayed_fidelity_2023, 
Inesta_performance_2023, 
Davies_entanglement_2024,
Collins_multiplexed_2007,
Simon_quantum_2007}, 
which divide a long-distance channel into many elementary links, use pairwise entanglement swapping to distribute entanglement between the two end nodes, and purify entanglement when necessary. This usually requires the temporary storage and on-demand retrieval of the photon-transmitted quantum information in the long-lived matter state of a quantum memory \cite{Kuzmich.etal2003, Julsgaard.etal2004, Lvovsky.etal2009b, Hedges.etal2010, Heshami.etal2016, Lan_multiplexed_2009, Zhang_experimental_2016, Kutluer_solid-state_2017, Wen_multiplexed_2019, Li_multicell_2021}. 
A major challenge in this regard is the limited memory efficiency and fidelity \cite{Grosshans.etal2001, Varnava.etal2006} and the requirement of highly synchronized operations \cite{humphreys_deterministic_2018, LiuXiao.etal2021}.

An alternative approach is the dissipative generation of entanglement, proposed and studied on several platforms \cite{plenio_cavity-loss-induced_1999, kastoryano_dissipative_2011, krauter_entanglement_2011, rao_dark_2013,  cormick_dissipative_2013, morigi_dissipative_2015, lin_dissipative_2013, shankar_autonomously_2013, rao_dissipative_2017, cole_resource-efficient_2022,malinowski_generation_2022}. It promises a deterministic output that is independent of the initial state and resilient to errors. However, given that it requires driving fields and environmental damping channels to act globally on the whole system, it is not clear how to realize this approach remotely over a quantum network.

In this paper, we discuss a potential primitive for remote entanglement generation: entanglement buffers.
Similar to dissipative generation, entanglement buffers can accumulate entanglement in a way that is resilient to errors and largely independent of the buffer’s initial state, making them especially suitable for quantum networks.
Given a generic quantum network scenario in which a local source broadcasts a stream of maximally entangled states to two or more remote network parties for a given task, entanglement buffers accumulate the provided entanglement among the parties in local multi-qubit storage systems for later on-demand use. 
Crucially, they do not rely on the single-shot memory storage of a quantum state with maximum possible fidelity. 
Instead, the goal is to fill the buffer with, \textit{i.e.}, ``cache'', a maximum (required) amount of entanglement between the parties at the expense of as many source state copies as needed.
This is in contrast to works \cite{Briegel.etal1998, 
Dur.etal1999, 
Sangouard.etal2011, 
bernardes_rate_2011, 
sheng_hybrid_2013, 
behera_demonstration_2019, 
pouryousef_quantum_2022, 
Inesta_performance_2023, 
elsayed_fidelity_2023,
Davies_entanglement_2024,
Collins_multiplexed_2007,
Simon_quantum_2007} that explore the operation strategies on multiple independent quantum memories based on the average fidelity with respect to a target Bell state.
They design and optimize strategies from a quantum repeater perspective by generating and purifying stored entangled states, including a recent work that introduces the notion of ``entanglement buffering'' \cite{Davies_entanglement_2024} in this context.
In that setting, two pairs of quantum memories are used, where one pair stores an incoming Bell pair that is later consumed to purify the other via entanglement pumping.
Here, we consider an entanglement buffer consisting of an ensemble of qubits per party that buffers Bell-state entanglement through a sequence of SWAP interactions.
We quantify the overall amount of stored entanglement in terms of the \textit{logarithmic negativity} $E$ \cite{Vidal.etal2002, plenio_logarithmic_2005}, which offers an alternative figure of merit to fidelity, particularly when it comes to buffering a certain amount of entanglement rather than storing specific states.

In the following, we will introduce the model for an entanglement buffer in a bipartite setting using partial swap operations for the caching of Bell-pair entanglement, as depicted in Fig.~\ref{fig:model}. We first assess the entanglement caching of a single Bell pair, highlighting the influence of the initial buffer state, the interaction parameters, and the buffer size. We then focus on the practical benefits of an entanglement buffer in a quantum network scenario, studying how to fill the buffer with a steady amount of entanglement from a stream of Bell pairs, and benchmarking its performance when the interaction suffers from imperfect control, noise, and transmission loss. We close with a conclusion and outlook.

\begin{figure}[tb]
  \includegraphics[width=.99\columnwidth]{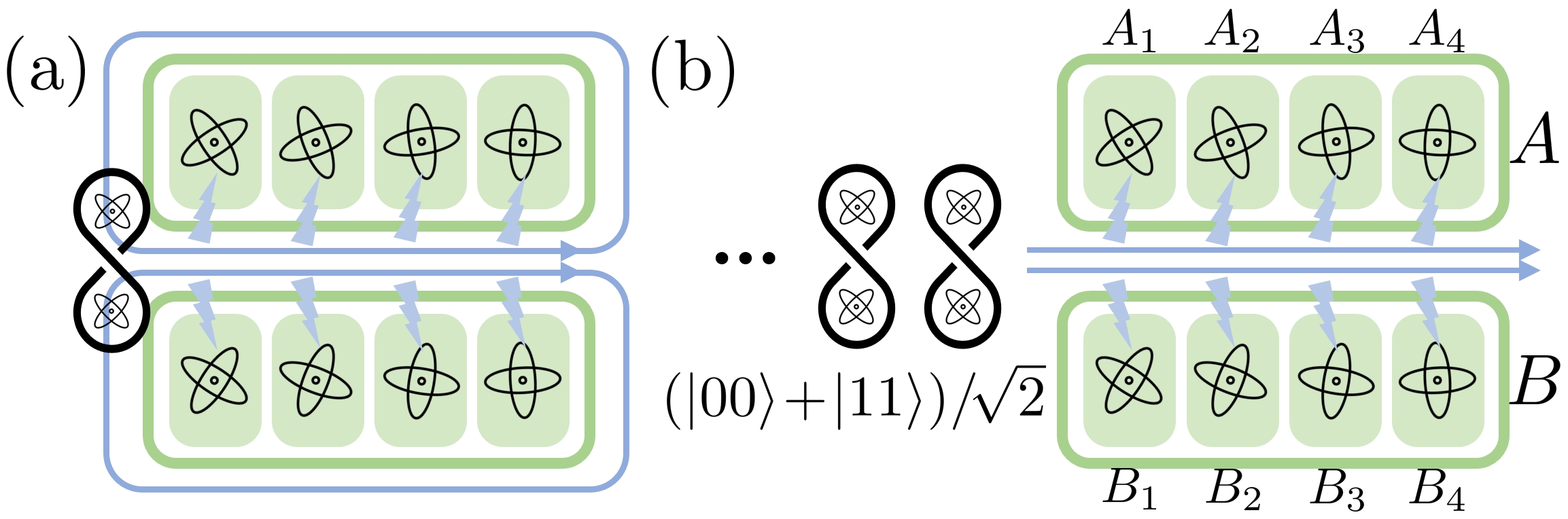}
  \caption{Sketch of a multi-pair entanglement buffer caching Bell-pair entanglement through a sequence of partial swap operations between the subsequent qubit pairs of the buffer and either (a) a single Bell pair over one or more repetitions or (b) a sequence of Bell pair copies. Entanglement is shared between two distant parties $A$ (upper part) and $B$ (lower part).}
  \label{fig:model}
\end{figure}

\textit{Scheme.---}
Consider the bipartite model setting depicted in Fig.~\ref{fig:model}: Alice and Bob share a $k$-pair entanglement buffer, each comprised of $k$ qubits on their respective sides. 
A source provides maximally entangled qubit pairs in the Bell state $|\psi_1\ra = (|00\ra + |11\ra)/\sqrt{2}$. 
The objective is to cache the entanglement of each source unit and successively fill the buffer. 
To this end, each incoming source unit undergoes the same uniform caching operation, in which its two qubits interact sequentially with all $k$ buffer qubits on Alice's and Bob's sides, respectively. 
This design avoids the need for selective addressing between memory elements, as commonly employed in repeater-type protocols \cite{Collins_multiplexed_2007, Simon_quantum_2007, Lan_multiplexed_2009, Zhang_experimental_2016, Kutluer_solid-state_2017, Wen_multiplexed_2019, Li_multicell_2021}. 
Instead, it enables entanglement accumulation through passive, symmetric operations—making the scheme particularly suited to platforms such as trapped ions, superconducting qubits, or photonic circuits, where uniform or sequential couplings are more naturally implemented than dynamically routing or selectively activating individual memories. 
In such settings, our protocol reduces experimental complexity and is designed to tolerate imperfect control, noise, and transmission loss, as will be analyzed in the following.

It is evident that for both our scheme and the memory-based quantum repeater scheme, achieving optimality relies upon the realization of a perfect full SWAP gate between incoming qubits and storage qubits. However, reality is rife with various types of noise and imperfections that inevitably impact the SWAP interaction, thereby limiting the effectiveness of the memory-based quantum repeater scheme. Nevertheless, we will demonstrate in the subsequent discussion that the entanglement buffer scheme adeptly accommodates imperfect SWAP interactions for entanglement caching.

Theoretically, an arbitrary two-qubit unitary operator can be described in the form of a Heisenberg exchange interaction $U=U_A\otimes U_B\exp(-\I \sum_{i=1}^3 r_i\sigma_i\otimes\sigma_i)V_A\otimes V_B$, 
where $\sigma_{1,2,3}$ are Pauli matrices and $U_{A,B}$ and $V_{A,B}$ are local unitaries \cite{kraus_optimal_2001, leifer_optimal_2003, zhang_geometric_2003}. The perfect SWAP gate can then be realized with interaction parameter $r_{1,2,3}=\pi/2$, which can be tuned with interaction strength and time. 
More generally, for $r_1=r_2=r_3$, it can achieve the partial SWAP gate \cite{leifer_optimal_2003, fan_optimal_2005}.

Utilizing such an interaction model, a SWAP gate can be realized in various experimental platforms, such as nuclear magnetic resonance (NMR) \cite{ichikawa_minimal_2013}, superconductor devices \cite{roth_analysis_2017, nguyen_programmable_2024}, and trapped ions \cite{molmer_multiparticle_1999, debnath_demonstration_2016, gan_hybrid_2020, Drmota_robust_2023}. 
When SWAP gates are implemented in these physical systems, the strengths $r_i$ of the Heisenberg interaction terms cannot be tuned independently. Experimental control typically encompasses two components: the XY interaction (where $r_1 = r_2$ with tunable $r_1$) and the ZZ interaction (with tunable $r_3$). We therefore focus our view on SWAP interactions characterized by the parametric form
\begin{eqnarray}\label{eq:S}
    \oS(\alpha,\beta) &=& \E^{-\I\beta(|01\ra\la 01| + |10\ra\la 10|)/2} [\mathrm{SWAP}]^{\alpha/\pi} \\
    &=& \E^{\I(\alpha-\beta)(|01\ra\la 01| + |10\ra\la 10|)/2} \E^{-\I\alpha (|10\ra\la 01| + |01\ra\la 10|)/2}\,, \nonumber
\end{eqnarray}
where the swap angle $\alpha \in[0,\pi]$ and relative phase $\beta \in [0,\alpha]$ are corresponding to XY and ZZ interactions, respectively. 
A perfect full SWAP is given by $\oS(\pi,0)$, and a full iSWAP by $\oS(\pi,\pi)$, which is sometimes easier to realize \cite{ichikawa_minimal_2013, roth_analysis_2017}.
In practice, a partial iSWAP $\oS(\alpha,\alpha)$ is often easier to realize, such as in trapped ion systems \cite{molmer_multiparticle_1999, debnath_demonstration_2016, Drmota_robust_2023}, because it is directly generated by a Hamiltonian of the form $\oH \propto |01\ra\la 10| + |10\ra\la 01|= ( \os_x\otimes\os_x + \os_y\otimes\os_y )/2$.
See Appendix~\ref{Appx:interaction} for more details \cite{supp}.

The initially empty buffer shall assume a product state, $\rho^{(2k)}_0 = \rho^{(k)}_0 \otimes \rho^{(k)}_0$ with $E=0$.
Labelling the two source qubits by $A,B$ and the $j$-th buffer qubit pairs by $A_j,B_j$, the total caching unitary applied to the combined source-buffer state reads as
\begin{equation}\label{eq:U}
    \oU (\alpha,\beta) = \prod_{j=1}^k \oS_{AA_j} (\alpha,\beta) \otimes \oS_{BB_j} (\alpha,\beta).
\end{equation}
As discussed earlier, the ideal caching scheme consists in performing a perfect full swap, which would fill the buffer with $k$ ebits of entanglement from $k$ source units, regardless of the initial (unentangled) buffer state.
However, in practice, implementing full SWAPs is often challenging due to experimental limitations on source-buffer coupling strength, gate duration, or coherence time.
In many platforms, such as superconducting qubits or trapped ions, only partial two-qubit gates are directly available, making partial SWAP-type interactions a more realistic and tunable alternative for experimental realization.
One could remedy this by allowing each source unit to pass through the buffer multiple times, as sketched in Fig.~\ref{fig:model}(a), but in general, the amount of cachable entanglement from a single copy will be sensitive to the coupling parameters $(\alpha,\beta)$, the initial buffer state, and the buffer size. We will discuss this next, before introducing the multi-copy caching protocol shown in Fig.~\ref{fig:model}(b) that alleviates the demanding initial-state parameter requirements for filling the buffer.

\textit{The single-copy protocol.---}
We begin with the simplest, analytically tractable case of a single source unit, a single caching operation $\oU$, and a buffer of size $k=1$.
In practice, it is often reasonable to consider reusing the same Bell pair for multiple interactions with the buffer, especially when the single-pass coupling strength $\alpha$ is weak.
However, due to the composition rule $\oS(\alpha_1,\beta_1)\oS(\alpha_2,\beta_2) = \oS(\alpha_1+\alpha_2,\beta_1+\beta_2)$, such repeated applications are effectively equivalent to a single interaction with a stronger effective coupling.
Thus, by analyzing the single-pass case, we capture both the single-use and repeated-use scenarios within the same framework.

The state of the buffer after the caching operation is
\begin{equation}\label{eq:channelC}
    \rho^{(2)}_1 = \tr_{AB} \left[ \oU(\alpha,\beta) \rho^{(2)}_0 \otimes |\psi_1\ra\la \psi_1| \oU^\da (\alpha,\beta) \right] \equiv \cC [\rho^{(2)}_0];
\end{equation}
see Appendix~\ref{Appx:one_pair} \cite{supp} for the Kraus representation of the so defined quantum channel $\cC$. We distinguish the two opposite cases of a buffer initialized in the maximally mixed state, $\rho_0^{(2)} = \openone/2 \otimes \openone/2$, or in a pure state $\rho_0^{(2)} = |\phi_{\theta,\delta}\ra\la \phi_{\theta,\delta}|^{\otimes 2}$ with $|\phi_{\theta,\delta}\ra=\cos\theta |0\ra+\E^{\I \delta}\sin\theta |1\ra$. 

In the former case, $\cC$ turns the buffer state into a mixture of two Bell states and the maximally mixed state,
\begin{equation}\label{eq:rho1_mm}
    \rho_1^{(2)} = \frac{1-a^2}{4} \openone + \frac{a}{2} \left[(a+b)|\psi_1\ra\la \psi_1| + (a-b)|\psi_2\ra\la\psi_2|\right],
\end{equation}
with $a=\sin^2(\alpha/2)$, $b=\sin^2[(\alpha-\beta)/2]$, and $|\psi_2\ra = (|00\ra - |11\ra)/\sqrt{2}$. 
One can show that the cached amount of entanglement is given by 
\begin{equation}\label{eq:rho1_mm_E}
    E = \log_2 \left[ 1 - \frac{1-a(a+2b)-|1-a(a+2b)|}{4} \right],
\end{equation}
which is nonzero for $a(a+2b)>1$. Partial SWAPs ($b=a$) can fill the buffer if the coupling strength exceeds $\alpha > 2\arcsin(3^{-1/4}) \approx 0.55\pi$, with the maximum $E=1$ at $\alpha=\pi$. With growing phase angle $\beta$, the caching performance for any given $\alpha$ deteriorates, until $b=0$: partial iSWAPs are not able to cache any entanglement regardless of $\alpha$. 

\begin{figure}[tb]
  \includegraphics[width=0.99\linewidth]{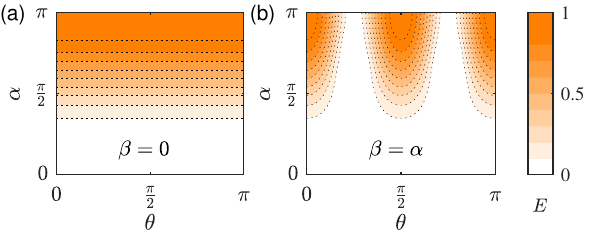}
  \caption{Entanglement cached from a single Bell pair in a 1-pair buffer by means of (a) a partial SWAP unitary and (b) a partial iSWAP of varying angles $\alpha$. Both buffer qubits are initialized in a pure state $\cos\theta\ket{0}+\sin\theta\ket{1}$ of varying $\theta$. The entanglement is measured in ebits of logarithmic negativity.
  }
  \label{fig:singlecopy}
\end{figure}

In the case of a pure initial buffer state, the adverse effect of $\beta \neq 0$ is less severe. 
The amount of cached entanglement mainly depends on the mixing angle $\theta$, but only weakly on the phase $\delta$; see Appendix~\ref{Appx:single} \cite{supp}. Setting $\delta=0$ for simplicity, we plot the cached entanglement as a function of $\theta$ and the swap angle $\alpha$ in Fig.~\ref{fig:singlecopy}, comparing (a) a partial SWAP operation to (b) a partial iSWAP. The latter is sensitive to the initial $\theta$, whereas the former is not. A full iSWAP, in particular, caches $E=\log_2 [1+\cos^2(2\theta)]$, which is zero for initial states on the equator of the Bloch sphere, $\theta = \pi/4, 3\pi/4$. In general, the buffer performs optimally and independently of $\beta$ if it is initialized in $|00\ra$ or $|11\ra$, yielding $E=\log_2 [1+\sin^4 (\alpha/2)]$. In fact, for arbitrary initial product states of the buffer, we always observe that the caching performance is optimal at $\beta=0$ (SWAP), deteriorates with $\beta$, and is worst at $\beta=\alpha$ (iSWAP).

\begin{figure}[tb]
  \includegraphics[width=0.99\linewidth]{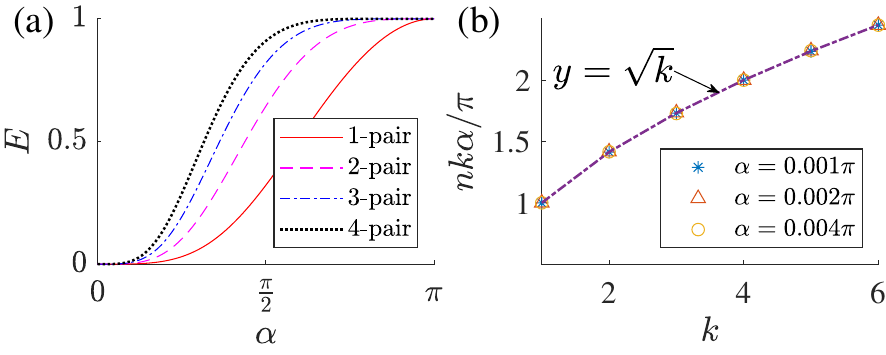}
  \caption{(a) Entanglement buffering from a single caching unitary of SWAP type with a single Bell pair for buffer sizes up to $k=4$. We plot the cached $E$ ebits as a function of the swap angle. (b) Number $n$ of repeated caching operations required to reach $E=1$ as a function of buffer size $k$, given fixed weak swap angles $\alpha$. We plot $n$ in terms of the total operation time $\propto nk\alpha$, which grows like $\sqrt{k}$ (dashed line). All buffer qubits are initialized in $|0\ra$.}
  \label{fig:onecopy_big}
\end{figure}

Our assessment of the single-pair buffer demonstrates that the caching of one ebit of entanglement from a single source unit depends crucially on the interaction strength and the initial buffer state. Before we move on to alleviate this dependence in a multi-copy protocol, we briefly discuss the influence of larger buffer sizes and numbers of caching cycles for a single source pair.
For larger $k$-pair buffers, subsequent caching operations no longer commute or combine to a single operation in general, as they consist of multiple partial swaps that distribute correlations between all the buffer qubit pairs. The benefit lies in the caching of more entanglement at smaller swap angles $\alpha$, as shown in Fig.~\ref{fig:onecopy_big}(a) up to $k=4$. However, if multiple caching operations are used to extract the entanglement from a single source copy, the total caching time will increase with the buffer size. Given a fixed interaction strength, the time per caching operation is proportional to $k\alpha$. In Fig.~\ref{fig:onecopy_big}(b), we show the total time of $n$ operations required to cache the 1 ebit from a single source unit, for various buffer sizes in the limit of weak swap angles $\alpha \ll \pi$. The time grows like $\sqrt{k}$; see Appendix~\ref{Appx:single} for the proof \cite{supp}. 
A similar behaviour was found in recent work \cite{mondal_local_2023}, which is a specific case of our single-copy protocol limited with $\beta=0$.

\textit{The multi-copy protocol.---}
Now we consider protocols for filling the buffer by consuming an arbitrarily long sequence of source Bell-pair units. Each unit is used in a single caching operation and then discarded, transforming the buffer state according to the channel $\cC$ in \eqref{eq:channelC}. With repeated applications of this channel, the buffer state undergoes a Markov chain, $\rho_n^{(2k)} = \cC \rho_{n-1}^{(2k)} = \cC^{\circ n} \rho_0^{(2k)}$ \cite{Ciccarello.etal2022, Raffaele.etal2023}. Expanding the channel $\cC$ into Kraus operators and employing the notation from Refs.~\cite{Raffaele.etal2023}, one can define the generator $\cL \rho := \cC [\rho] - \rho$ and bring it to Lindblad form. 

After sufficiently many steps, the buffer will generally converge towards a steady state $\rho_\infty^{(2k)}$, which obeys $\cL \rho_\infty^{(2k)} = 0$. We observe that this steady state is unique for $\alpha \neq 0$, given a symmetric initial product state of the buffer. In the case $k=1$, the state can be given analytically in terms of the Bell basis, 
\begin{eqnarray}
    \rho_\infty^{(2)}=\frac{1}{\mathcal{N}}\bigl\{
    &&b |\psi_1\ra\la \psi_1| +\left[ (a-b)^2 + b(1-b) \right] |\psi_2\ra\la \psi_2|\nonumber\\
    &&+(1-a)b \bigl(|\psi_3\ra\la \psi_3|+|\psi_4\ra\la \psi_4|\bigr)\Bigr\}\,,
    \label{eq:SS_1pair}
\end{eqnarray}
with $|\psi_{3,4} \ra = (|01\ra \pm |10\ra)/\sqrt{2}$ and $\mathcal{N} = a^2 + 4b(1-a)$;  see Appendix~\ref{Appx:multi} \cite{supp}. The corresponding logarithmic negativity is 
\begin{equation}\label{eq:E_SS_1pair}
    E_\infty = \log_2 \left\{ 1 + \frac{|a^2-2b|+|a^2-4ab+2b|-4b(1-a)}{2[a^2+4b(1-a)]} \right\}
\end{equation}
Once again, a full SWAP ($a,b=1$) results in the maximally entangled steady state $|\psi_1\ra$, reached after a single step. However, for other angle parameters $\alpha\neq 0$, and contrary to the single-copy case, partial iSWAPs ($b=0$) now outperform partial SWAPs: Regardless of the actual $\alpha$-value, the buffer \emph{always} converges towards the maximally entangled steady state $|\psi_2\ra$ under iSWAPs. Partial SWAPs, on the other hand, buffer entanglement only for $\alpha > 2 \arcsin(\sqrt{2/3}) \approx 0.61\pi$.

\begin{figure}
  \includegraphics[width=0.99\linewidth]{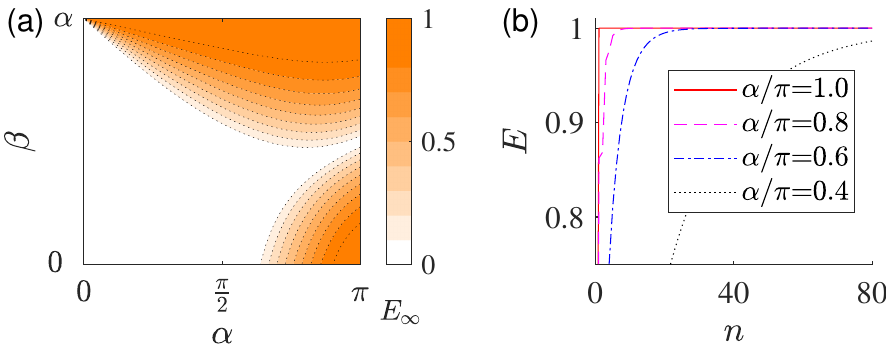}
  \caption{(a) Steady-state entanglement $E_\infty$ (ebits) cached from a sequence of Bell pairs in a 1-pair buffer, using a caching unitary of varying interaction parameters $\alpha,\beta$. (b) Entanglement $E$ accumulated in a 1-pair buffer over $n$ subsequent iSWAP caching steps ($\beta=\alpha$) with independent source pairs, comparing various $\alpha$-values.}
  \label{fig:multicopy}
\end{figure}

The steady-state entanglement \eqref{eq:E_SS_1pair} is plotted as a function of $\alpha,\beta$ in Fig.~\ref{fig:multicopy}(a), showing the iSWAP sweet spot close to $\beta=\alpha$. In (b), we show how the buffer fills with partial iSWAP steps of different $\alpha$, assuming the initial state $|00\ra$. 

\begin{figure}
  \includegraphics[width=0.99\linewidth]{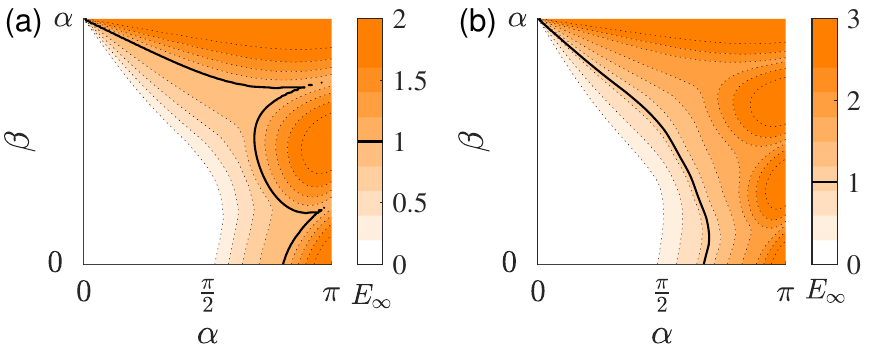}
  \caption{Steady-state entanglement $E_\infty$ cached by multi-copy protocols with varying interaction parameters $\alpha,\beta$, for (a) a $2$-pair and (b) a $3$-pair buffer. The solid line marks the parameter range (right of it) for which the buffer contains more than 1\,ebit.
    }
  \label{fig:multicopy_big}
\end{figure}

Larger buffers improve the caching performance as each source unit transfers its entanglement to  $k>1$ successive qubit pairs. Figure \ref{fig:multicopy_big} shows the steady-state entanglement as a function of $\alpha,\beta$ for (a) $k=2$ and (b) $k=3$. The range of interaction parameters for which a steady value of $E_\infty \geq 1\,$ebit can be stored (shaded area to the right of the solid line) grows with $k$, favoring the iSWAP ($\beta \to \alpha$) over the SWAP regime ($\beta \to 0$).

\begin{figure}
\includegraphics[width=1.00\linewidth]{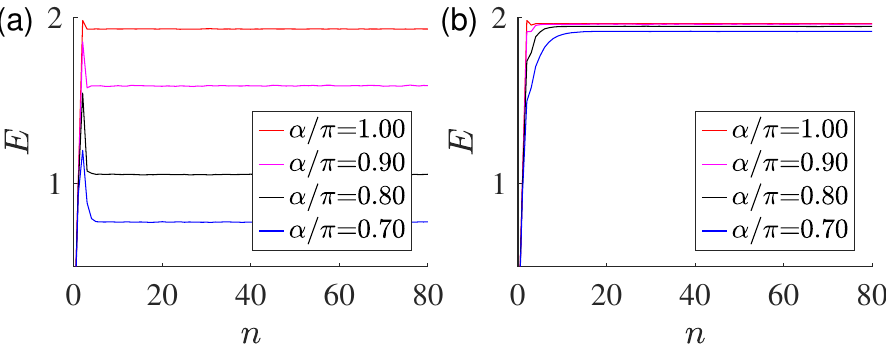}
  \caption{Average entanglement $E$ (ebits) of a $2$-pair buffer accumulated over $n$ steps subject to Gaussian-distributed fluctuations of the swap angle $\alpha$ and relative phase $\beta$ for (a) partial SWAPs ($\Tilde{\alpha} \sim \mathcal{N}(\alpha,\sigma)$, $\Tilde{\beta} \sim \mathcal{N}(0,\sigma)$) and (b) partial iSWAPs ($\Tilde{\alpha},\Tilde{\beta} \sim \mathcal{N}(\alpha,\sigma)$), with various mean $\alpha$-values. The standard deviation is $\sigma = 0.05\times \pi$ in each step $n$; we average over $20000$ random samples of $80$-step trajectories.}
  \label{fig:fluctuation}
\end{figure}

\textit{Robustness against imperfect control.---}
One advantage of the multi-copy protocol lies in its resilience to imperfect control of the interaction, i.e.,  systematic deviations or random fluctuations of the interaction parameters induced by miscalibration, drift, or noise. As previously discussed, the multi-copy protocol can attain an entangled steady state for a wide range of $\alpha$ and $\beta$ values, demonstrating robustness to parameter deviations that are not too large. For partial iSWAP interactions with fixed $\beta=\alpha$, the cached entanglement, reaching maximum, does not even depend on the $\alpha$-value.

We can therefore expect that, as one performs successive caching steps, random fluctuations of the interaction parameters in each step will not accumulate to a total loss of entanglement. In Fig.~\ref{fig:fluctuation}, we show the buffered entanglement of a $2$-pair buffer as a function of the number of caching steps, averaged over 20000 random samples for (a) partial SWAPs and (b) iSWAPs at various mean $\alpha$ values. Per step and sample, the interaction parameters are drawn from a Gaussian distribution with a standard deviation of $0.05 \times \pi$ around the mean values $\alpha \in[0,\pi]$ and $\beta\in[0,\alpha]$. We observe that smaller $\alpha$ values result in a more pronounced average loss of entanglement due to the fluctuations. At the same time, the advantage of partial iSWAP interactions remains as they can sustain more entanglement under fluctuations at smaller $\alpha$. See Appendix~\ref{Appx:adv} for more details \cite{supp}.

\begin{figure}
\includegraphics[width=1.00\linewidth]{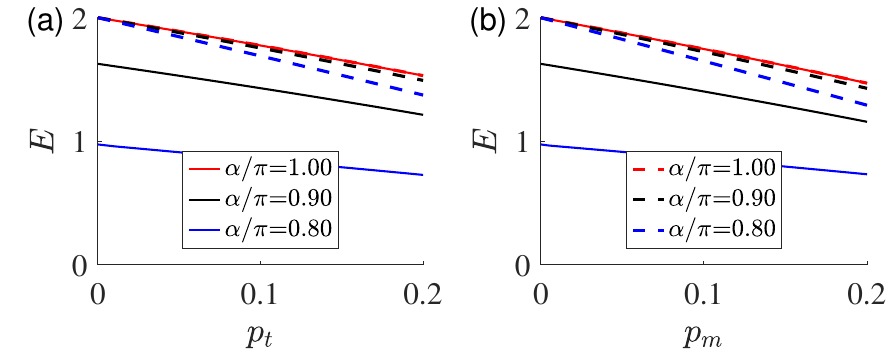}
  \caption{Steady-state entanglement $E_{\infty}$ (ebits) of a $2$-pair buffer subject to noise during (a) entanglement transmission and (b) entanglement storage for partial SWAPs (solid lines) and iSWAPs (dashed lines) with $\alpha/\pi=1.0,0.9,0.8$ (top to bottom; red, black, and blue, respectively). Notice that the red solid and dashed lines are almost indistinguishable. The steady results are estimated by $100$-step trajectories. }
  \label{fig:noise}
\end{figure}

\textit{Robustness against noise.---}
Entanglement distribution and buffering between remote parties in a quantum network can be subject to noise during entanglement transmission and during storage. We model both cases in terms of white noise, or identity noise, which is common and reasonable in realistic situations \cite{urbanek_mitigating_2021, mi_information_2021, foldager_can_2023, dalzell_random_2024}.

Noise during transmission (or source noise) can be accounted for by replacing the pure Bell state of the source pair with the mixture $(1-p_t)\ket{\psi_1}\bra{\psi_1}+p_t\openone/4$, which provides $E_t=\log_2(2-3p_t/2)$ ebits (for not too large noise $p_t \leq 2/3$). The strength of the noise $p_t$ is influenced by factors such as the quality of the source, the quality of the transmission channel, and the distance or time of transmission. 
In Fig.~\ref{fig:noise}(a), we plot the steady-state entanglement of a $2$-pair buffer (approximated by performing 100 caching steps) for SWAPs (solid) and iSWAPs (dashed) at varying $p_t$ and for three different swap angles $\alpha$. Once again, partial iSWAP operations exhibit greater tolerance to transmission noise compared to partial SWAP interactions.

We model noise during storage (memory noise) by partially depolarizing the total buffer state upon each caching step, according to the Markov chain  $\Tilde{\rho}_n^{(2k)} = (1-p_m)\cC \Tilde{\rho}_{n-1}^{(2k)} + p_m\openone/4^k$. 
Similar to transmission noise, the multi-copy buffering protocol under memory noise can still achieve steady entanglement caching with only finite loss, as shown in Fig.~\ref{fig:noise}(b). 
see Appendix~\ref{Appx:adv} for more details \cite{supp}.

\textit{Robustness against transmission loss.---}
When the Bell pairs are provided by a distant source over a quantum network, transmission loss is typically the main cause of failure. 
In our model, each transmitted Bell pair may either successfully reach both parties, be lost partially (single-sided loss), or be completely lost (two-sided loss). Crucially, the caching protocol operates passively: if one qubit is lost, the corresponding side performs no operation, while the other side independently executes its caching step. Thus, two-sided loss events leave the buffer unchanged. No adaptive feedback or heralding is required, which significantly simplifies practical implementation.

Within this framework, the advantage of multi-copy over single-shot entanglement storage becomes immediately apparent. Given equal transmission probability $p$ on both sides, successful single-shot transfer happens with probability $p^2$, while multi-copy caching naturally accommodates two-sided loss events occurring with probability $(1-p)^2$. Single-sided loss events, however, always degrade the buffer state.
For example, given full SWAP steps ($\alpha=\pi$), single-sided loss and successful transfer will change the state of 1-pair buffer between $\ket{\psi_1}\bra{\psi_1}$ and $\openone/2\otimes\openone/2$, and two-sided loss never change the states. Then the probability to find a 1-pair buffer filled with 1 ebit after $n$ caching steps is 
\begin{equation}
    q_n = p^2 \sum_{\ell=0}^{n-1} (1-p)^{2\ell} = \frac{1-(1-p)^{2n}}{2-p}p \xrightarrow{n\gg 1} \frac{p}{2-p}, \label{eq:qn_SWAP}
\end{equation}
with $q_n > p^2$ for any $n\geq 2$. 

As a figure of merit for arbitrary $\alpha$, let $q_n(E)$ be the probability to find the buffer in a state of at least  $E$ ebits of entanglement, out of all possible caching trajectories consuming $n$ source units at a given $p<1$. We estimate this by the relative frequency taken from $M=5000$ random samples of $n$-step trajectories (we choose $n=10$ to let $(1-p)^{2n}\leq 10^{-6}$). 

\begin{figure}
  \includegraphics[width=1.00\linewidth]{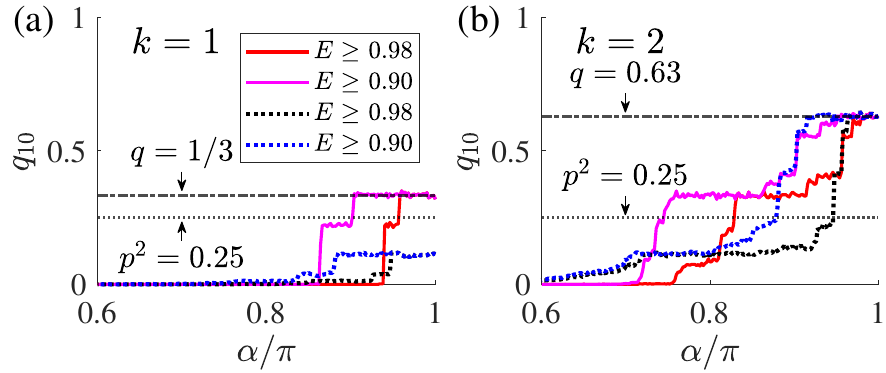}
  \caption{Probability to cache $E$ ebits of entanglement by partial SWAPs (solid lines) and iSWAPs (dash lines) of varying swap angle $\alpha$ with a sequence of Bell pairs. We consider a single-side transmission probability $p=0.5$ for (a) a $1$-pair buffer and (b) a $2$-pair buffer. The horizontal dotted line marks the probability to successfully transmit the Bell pair, the dash-dotted line gives the asymptotic probability to cache 1 ebit in an arbitrarily long sequence at $\alpha=\pi$.}
  \label{fig:probability}
\end{figure}

Our results for $q_{10}(E)$ are shown in Fig.~\ref{fig:probability}(a) and (b) for a 1-pair and a 2-pair buffer, respectively. We compare the partial SWAP (solid lines) and the iSWAP case (dash lines) as a function of the swap angle $\alpha$, for $p\!=\!0.5$ and two values of $E\lesssim 1\,$ebit each. The dash-dotted line marks the optimal 1-ebit caching probability for an asymptotic sequence of full SWAPs; one can beat the $p^2$-threshold and approach this optimum for a moderate range of $\alpha$-values deviating from $\pi$. The iSWAP case, however, is more sensitive to transmission loss: in order to refill a fully depleted buffer of maximally mixed qubit pairs, one needs a longer sequence of successful iSWAPs than SWAPs, so that single-sided loss events are more detrimental to iSWAP caching. For example, one full SWAP fills a 1-pair buffer upon success, or resets it to the maximally mixed state $\openone \otimes \openone/4$ upon single-sided loss. Whereas, two subsequent full iSWAPs are required to fill the maximally mixed buffer and each single-sided failure leads to a reset. Consequently, the 1-ebit success rate is even lower than $p^2$, as seen in Fig.~\ref{fig:probability}(a). For the 2-pair buffer, SWAPs and iSWAPs can reach the same 1-ebit success probability of about $63\%$, and greater buffer sizes can be more beneficial.
In this example, the improvement, from $p^2 = 25\%$ to $63\%$, is equivalent to tripling the transmission distance in fiber-based networks, where the success probability decays exponentially with length \cite{Gisin_quantum_2002}.
Notably, this performance enhancement persists even when additional imperfections, such as memory noise (decoherence), are taken into account; see Appendix~\ref{Appx:adv} for more details \cite{supp}.


\textit{Conclusion.---}
We have introduced the concept of entanglement buffers, local single- or multi-qubit devices shared by two parties, which accumulate and temporarily store (``cache'') entanglement provided by a source for later on-demand use. Ideally, the buffer could be filled by means of full SWAP operations with source Bell pairs, thereby acting as a quantum memory for Bell states. Crucially, however, one or more ebits of entanglement can also be cached by consuming multiple source copies when only weaker partial SWAP or iSWAP interactions are available. This could facilitate entanglement storage for quantum networking tasks in a wider range of realistic settings with limited control accuracy.

We have also demonstrated that the entanglement buffer scheme can tolerate realistic imperfections such as systematic or random deviations in the interaction parameters and decoherence during transmission or storage.
Moreover, we have analyzed the impact of transmission loss and found that entanglement buffers can have a higher success rate than single-shot Bell state storage. 

Future works could explore the extraction and concentration of buffered entanglement onto selected qubit pairs, enabling its direct use in quantum communication, computation, or sensing. 
Potential approaches include local SWAP operations, entanglement distillation, or routing protocols, depending on the platform \cite{Bennett_purification_1996, Dur_entanglement_2007, Perseguers_quantum_2010, Azuma_all_2015, Campbell_roads_2017}. 
Another direction is to extend the buffering protocol to support genuine multipartite entanglement and its applications in distributed quantum tasks. 
In addition, entanglement buffers could be integrated into repeater-based network architectures to reduce control overhead, increase resilience to channel noise, and support scalable passive entanglement distribution.
Finally, other experimental techniques such as error correction for known initial buffer states \cite{Drmota_robust_2023} could be employed to further enhance the buffering performance.

\acknowledgments
We are grateful to Satoya Imai and Yiru Zhou for interesting discussions.
This work was supported by the Deutsche Forschungsgemeinschaft (DFG, German Research Foundation, project numbers 447948357 and 440958198), the Sino-German Center for Research Promotion (Project M-0294), and the German Ministry of Education and Research (Project QuKuK, BMBF Grant No.~16KIS1618K).
Y.-C. Liu is also supported by the DFG Cluster of Excellence MATH+ (EXC-2046/1, Project No.~390685689) funded by the Deutsche Forschungsgemeinschaft (DFG).‌

%
%

\title{Supplemental Material: Entanglement Buffers}
\author{Ye-Chao Liu}
\affiliation{Naturwissenschaftlich-Technische Fakult{\"a}t, Universit{\"a}t Siegen, Walter-Flex-Stra{\ss}e 3, 57068 Siegen, Germany}

\author{Otfried G\"uhne}
\email{otfried.guehne@uni-siegen.de}
\affiliation{Naturwissenschaftlich-Technische Fakult{\"a}t, Universit{\"a}t Siegen, Walter-Flex-Stra{\ss}e 3, 57068 Siegen, Germany}

\author{Stefan Nimmrichter}
\email{stefan.nimmrichter@uni-siegen.de}
\affiliation{Naturwissenschaftlich-Technische Fakult{\"a}t, Universit{\"a}t Siegen, Walter-Flex-Stra{\ss}e 3, 57068 Siegen, Germany}

\date{\today}
\maketitle
%

\onecolumngrid

\appendix

\section{Disscussion of the SWAP-type interaction}\label{Appx:interaction}
We begin with a general two-qubit unitary gate $U$, which can be expressed in terms of a 
Heisenberg exchange interaction \cite{kraus_optimal_2001, leifer_optimal_2003, zhang_geometric_2003},
\begin{eqnarray}
    U \sim U_d = \exp(-\I \sum_{i=1}^3 r_i\sigma_i\otimes\sigma_i)\,,
\end{eqnarray}
such that $U=(U_A\otimes U_B) U_d (V_A\otimes V_B)$,
where $U_{A,B}$ and $V_{A,B}$ are local unitaries, and the Pauli matrices $\sigma_{1,2,3}$ are represented by 
\begin{eqnarray}
    \os_x=\begin{pmatrix}
        0&  1\\
        1&  0
    \end{pmatrix}\,,\quad
    \os_y=\begin{pmatrix}
        0&  \I\\
        -\I&  0
    \end{pmatrix}\,,\quad
    \os_z=\begin{pmatrix}
        1&  0\\
        0&  -1
    \end{pmatrix}\,,
\end{eqnarray}
in the single-qubit computational basis $(|0\ra ,|1\ra)$. 

As we discussed in the main text, the Heisenberg interaction terms are not mutually independent, but encompasses two components: the resonant flip-flop (or XY) interaction and the controlled phase rotation (or ZZ-interaction), i.e., 
\begin{eqnarray}
    &\oR_{xy} & = \exp\bigl[-\I r_1 (\sigma_x\otimes\sigma_x+\sigma_y\otimes\sigma_y) \bigr] \,,\\
    &\oR_{zz} & = \exp\bigl[-\I r_3 (\sigma_z\otimes\sigma_z) \bigr]\,,
\end{eqnarray}
so the experimental related interaction can be written as
\begin{eqnarray}
    U_{\rm exp} & = \oR_{xy}(r_1)\oR_{zz}(r_2)\,.
\end{eqnarray}
Then the full two-qubit SWAP and iSWAP gates,
\begin{eqnarray}
    \text{SWAP} &=& \ket{00}\bra{00}+\ket{11}\bra{11}+\ket{01}\bra{10}+\ket{10}\bra{01}\,,\\
    \text{iSWAP} &=& \ket{00}\bra{00}+\ket{11}\bra{11}-i(\ket{01}\bra{10}+\ket{10}\bra{01})\,,
\end{eqnarray}
can be realized as 
\begin{eqnarray}
    \text{SWAP}= e^{\I \pi/4} \oR_{xy}(r_1=\pi/4)\oR_{zz} (r_3=\pi/4)\,, \qquad \text{iSWAP}= \oR_{xy}(r_1=\pi/4).
\end{eqnarray}
Here we employ the iSWAP sign convention $(-i)$ for simplicity, which can realized experimental as well \cite{ichikawa_minimal_2013}. 
It relates to the usual one in the literature by hermitian conjugation. 

In order to further simplify our calculation, let $r_1=r_2=\alpha/4, r_3=-\beta/4$, and include the global phase of SWAP into the phase rotation (ZZ), 
\begin{eqnarray}
    \oR_{xy}(\alpha) &=& \exp\left[-\I \alpha (\os_x\otimes \os_x + \os_y \otimes \os_y)/4\right]\,,\\
    \oP(\beta) &=& \exp\left[\I \beta (\os_z \otimes \os_z-\openone\otimes\openone)/4 \right]\,.
\end{eqnarray}
Partial SWAPs and partial iSWAPs hence can be easily represented as 
\begin{eqnarray}
    [{\rm SWAP}]^{\alpha/\pi} = \oP(-\alpha)\oR_{xy}(\alpha)\,,\qquad
    [{\rm iSWAP}]^{\alpha/\pi} = \oR_{xy}(\alpha)\,,
\end{eqnarray}
and a general partial swap operation can be the form, as also stated in Eq.~\eqref{eq:S} in the main text,  
\begin{eqnarray}
    \oS(\alpha,\beta):=\oP(\beta)  [{\rm SWAP}]^{\alpha/\pi}= \oP(\beta-\alpha)\oR_{xy}(\alpha) =  \begin{pmatrix}
        1&      0&      0&      0\\
        0&      \E^{-\I \beta/2}&      0&      0\\
        0&      0&      \E^{-\I \beta/2}&      0\\
        0&      0&      0&      1
    \end{pmatrix}
    \begin{pmatrix}
        1&      0&      0&      0\\
        0&      \frac{1+\E^{\I\alpha}}{2}&      \frac{1-\E^{\I\alpha}}{2}&      0\\
        0&      \frac{1-\E^{\I\alpha}}{2}&      \frac{1+\E^{\I\alpha}}{2}&      0\\
        0&      0&      0&      1
    \end{pmatrix}
    \,,
\end{eqnarray}
where $\alpha\!\in\! [0,\pi]$ and $\beta\!\in\![0,\alpha]$. 
Due to the commutativity of $\oS(\alpha,0)$ and $\oP(\beta)$, we have the composition rule
\begin{eqnarray}\label{eq:appx_CompRule}
    \oS(\alpha_1,\beta_1)\oS(\alpha_2,\beta_2) = \oP(\beta_1) \oS(\alpha_1,0) \oP(\beta_2) \oS(\alpha_2,0) = \oP(\beta_1)  \oS(\alpha_2,0)  \oS(\alpha_1,0) \oP(\beta_2)= \oS(\alpha_1+\alpha_2,\beta_1+\beta_2)\,.
\end{eqnarray}
To better understand the parametrization of general swap-type gates, in Fig.~\ref{fig:Appx:fidelity}, we show the gate fidelities of $ \oS(\alpha,\beta)$ with respect to the ideal SWAP and iSWAP gates with varying parameters based on their Choi states.
For example, the partial SWAP and iSWAP gates with $\alpha=0.90\pi$ have a fidelity of $98.16\%$ and $98.77\%$ with respect to the ideal SWAP and iSWAP gates, respectively.

\begin{figure}[tb]
    \includegraphics[width=0.49\linewidth]{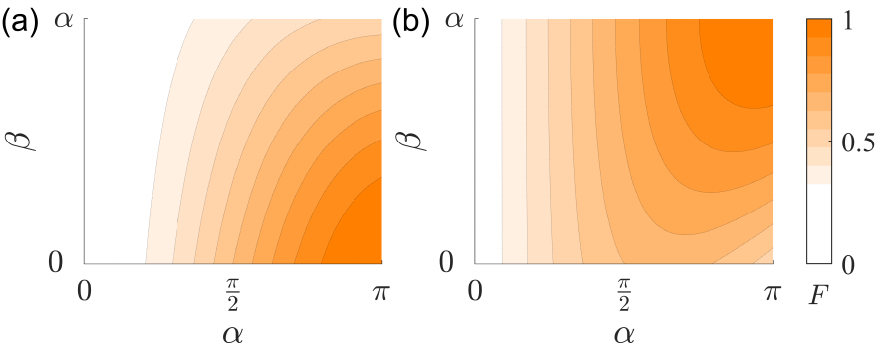}
    \caption{Fidelity of the general swap-type gate $\oS(\alpha,\beta)$ with respect to the ideal (a) SWAP and (b) iSWAP gates, as a function of the interaction parameters $\alpha$ and $\beta$. }
    \label{fig:Appx:fidelity}
\end{figure}

\section{1-pair caching transformation}\label{Appx:one_pair}
For a 1-pair buffer and a source Bell pair, the caching unitary is given by a tensor product of general swaps on the $A$-side and on the $B$-side, $\oU(\alpha,\beta) = \oS_{AA_1} (\alpha,\beta) \otimes \oS_{BB_1} (\alpha,\beta)$. 
Given that the source pair is in the Bell state $|\psi_1\ra$ and assuming the buffer is initialized in some $\rho_0^{(2)}$, application of the caching unitary transforms the reduced buffer state according to the quantum channel \eqref{eq:channelC},
\begin{equation}\label{eq:channelC_k1}
    \rho^{(2)}_1 = \cC [\rho^{(2)}_0] = \tr_{AB} \left[ \oU(\alpha,\beta) \rho^{(2)}_0 \otimes |\psi_1\ra\la \psi_1| \oU^\da (\alpha,\beta) \right] =\sum_{j=1}^4 \oM_j \rho_0^{(2)} \oM_j^\da, \qquad \text{with}\quad \oM_j = \la \psi_j|\oU(\alpha,\beta)|\psi_1\ra.
\end{equation}
Here, we have expanded the partial trace over the source pair $AB$ in the basis of two-qubit Bell states, $|\psi_{1,2}\ra = (|00\ra \pm |11\ra)/\sqrt{2}$ and $|\psi_{3,4}\ra = (|01\ra \pm |10\ra)/\sqrt{2}$. The resulting Kraus operators $\oM_j$, which act on the buffer qubit pair $A_1B_1$, can also be expanded in the Bell basis,
\begin{eqnarray} 
    \oM_1 
    &=& \frac{1}{4} \E^{-\I \beta} (1 + \E^{2\I \alpha} + 2\E^{\I \beta})\ket{\psi_1}\bra{\psi_1} 
    + \frac{1}{2}(1 + \E^{\I (\alpha-\beta)})\ket{\psi_2}\bra{\psi_2}
    + \frac{1}{2} \E^{-\I\beta/2} (1 + \E^{\I\alpha})(\ket{\psi_3}\bra{\psi_3} + \ket{\psi_4}\bra{\psi_4})\,,\nonumber\\
    \oM_2 
    &=& \frac{1}{2}(1 - \E^{\I (\alpha-\beta)})\ket{\psi_1}\bra{\psi_2}
    + \frac{1}{4} \E^{-\I \beta} (2\E^{\I \beta} - 1 - \E^{2\I \alpha})\ket{\psi_2}\bra{\psi_1}\,,\nonumber\\
    \oM_3 
    &=& \frac{1}{2} \E^{-\I\beta/2} (1 - \E^{\I\alpha})\ket{\psi_1}\bra{\psi_3} +\frac{1}{4} \E^{-\I \beta} (1 - \E^{2\I\alpha}) \ket{\psi_3}\bra{\psi_1}\,,\nonumber\\
    \oM_4 
    &=& \frac{1}{2} \E^{-\I\beta/2} (1 - \E^{\I\alpha})\ket{\psi_1}\bra{\psi_4} -\frac{1}{4} \E^{-\I \beta} (1 - \E^{2\I\alpha}) \ket{\psi_4}\bra{\psi_1}\,. \label{eq:channelC_k1_KrausOp}
\end{eqnarray}
For the simple case of an initially maximally mixed state $\rho_0^{(2)} = \openone/4$, we obtain
\begin{equation}\label{eq:channelC_k1_mm}
\begin{aligned}
    \rho^{(2)}_1 
    =\frac{1}{4}\sum_{j=1}^4 \oM_j \oM_j^\da
    &= \left\{\frac{1}{4}\sin^4(\alpha/2)+\frac{1}{2}\sin^2(\alpha/2)\sin^2[(\alpha-\beta)/2]+\frac{1}{4}\right\} &\ket{\psi_1}\bra{\psi_1}& \nonumber\\
    &+ \left\{\frac{1}{4}\sin^4(\alpha/2)-\frac{1}{2}\sin^2(\alpha/2)\sin^2[(\alpha-\beta)/2]+\frac{1}{4}\right\} &\ket{\psi_2}\bra{\psi_2}& \nonumber\\
    &+ \frac{1}{4}\left[1-\sin^4(\alpha/2)\right] &\bigl(\ket{\psi_3}\bra{\psi_3}& + \ket{\psi_4}\bra{\psi_4} \bigr).
\end{aligned}
\end{equation}
Expressed in terms of $a=\sin^2(\alpha/2)$ and $b=\sin^2[(\alpha-\beta)/2]$, this matches \eqref{eq:rho1_mm} and yields the logarithmic negativity \eqref{eq:rho1_mm_E} in the main text.

\section{Additional details on the single-copy protocol}\label{Appx:single}

Here we present several results concerning the entanglement caching of a single Bell-pair copy. First we introduce a spin representation of caching operations in the limit of weak interactions, $\alpha \to 0$. This will serve as an approximation to describe entanglement caching from a single source pair into $k$-pair buffers by means of many repeated applications of a weak caching unitary. For $k=1$, the description is exact at any $\alpha$.

Letting $\beta=r\alpha$ with $0\leq r \leq 1$ and making use of the notations in Appendix~\ref{Appx:one_pair}, the total caching unitary, Eq.~\eqref{eq:U} in the main text, can be expressed as
\begin{eqnarray}
    \oU(\alpha,\beta)&=&\prod_{j=1}^k 
    \E^{\I(1-r)\alpha/2}
    \E^{-\I(1-r)\alpha\oh_z^j}
    \E^{-\I\alpha\oh_{xy}^j}\,,\\
    \text{with}\qquad \oh_z^j &=& \frac{1}{4}(\os_z^A \os_z^{A_j} + \os_z^B \os_z^{B_j})\,,\qquad \oh_{xy}^j = \frac{1}{4}(\os_x^A \os_x^{A_j} + \os_y^A \os_y^{A_j} + \os_x^B \os_x^{B_j} + \os_y^B \os_y^{B_j})\,.
\end{eqnarray}
Here, we omit the tensor products for simplicity, and we instead denote the subspace each Pauli matrix acts on by a superscript: $A,B$ for the two source qubits, and $A_j,B_j$ for the $j$-th qubit pair of the buffer.
In the limit of small $\alpha$, we can expand the caching unitary with help of the Baker–Campbell–Hausdorff formula,
\begin{eqnarray}
    \oU(\alpha,\beta) = \E^{\I k(1-r)\alpha/2}
    \E^{-\I(1-r)\alpha\sum_j\oh_z^j}
    \E^{-\I\alpha\sum_j\oh_{xy}^j}\E^{\cO(\alpha^2)}\,,
\end{eqnarray}
omitting the second-order terms $\E^{\cO(\alpha^2)}$.
This allows us to express the operators acting on the $k$ buffer qubits on the $A$-side and on the $B$-side in terms of collective spin operators $\oJ^{A,B}$, with
\begin{eqnarray}
    \oJ_{x,y,z} = \frac{1}{2} \sum_{j=1}^k \os_{x,y,z}\,,\quad
    \oJ_{\pm} = \oJ_x \pm i \oJ_y = \sum_{j=1}^k \os_{\pm}\,,\quad
    \os_{\pm} = \frac{1}{2}(\os_x \pm i \os_y)\,.
\end{eqnarray}
The eigenstates $|m\ra$ of $\oJ_z$, with $m=-k/2,\ldots,k/2$, are the symmetric Dicke states; in particular $|-k/2\ra = |0\ra^{\otimes k}$ and $|k/2\ra = |1\ra^{\otimes k}$. The states obey 
\begin{eqnarray}
    \oJ_z \ket{m} = (-m) \ket{m}\,, \qquad \oJ_\pm \ket{m} = \sqrt{\frac{k}{2}\left(\frac{k}{2}+1\right)-m(m\pm 1)} \ket{m\pm 1}\,.
\end{eqnarray}
With this, the approximate weak-$\alpha$ caching unitary can be written as
\begin{eqnarray} \label{eq:U_weak}
    \oU(\alpha,\beta) 
    \simeq
    \E^{-\I (1-r)\alpha(\oH_z^A+\oH_z^B)}
    \E^{-\I\alpha( \oH_{xy}^A + \oH_{xy}^B)}\,,\qquad \text{with} \qquad 
    \oH_z=\frac{1}{2}\os_z \oJ_z \,, \qquad \oH_{xy} = \frac{1}{2}(\os_{+} \oJ_{-} + \os_{-} \oJ_{+})\,.
\end{eqnarray}
Here we have omitted the irrelevant global phase $\E^{\I k(1-r) \alpha/2}$.
The operator exponential with $\oH_{xy}$ can be expanded as
\begin{eqnarray}
    \E^{-\I\alpha(\oH_{xy}^A + \oH_{xy}^B)} =
     \left[\cos(|\oH_{xy}^A|\alpha)-i\alpha\sinc(|\oH_{xy}^A|\alpha)\oH_{xy}^A\right]
    \otimes
    \left[\cos(|\oH_{xy}^B|\alpha)-i\alpha\sinc(|\oH_{xy}^B|\alpha)\oH_{xy}^B\right]\,,
\end{eqnarray}
introducing the absolute-value operator
\begin{eqnarray}
    |\oH_{xy}|:=\sqrt{\oH_{xy}^2}=\frac{1}{2}\sqrt{\ket{1}\bra{1}\otimes\oJ_-\oJ_++\ket{0}\bra{0}\otimes\oJ_+\oJ_-}\,.
\end{eqnarray}
It satisfies $[H_{xy},|H_{xy}|]=0$ and is thus diagonal in the product basis of buffer and source $z$-eigenstates.

Suppose the buffer qubits are initialized in the state $|0\ra^{\otimes 2k}=|-k/2\ra|-k/2\ra$. Then a single Bell pair can at most induce a single Dicke excitation on each side. Abbreviating the Dicke states $|g\ra \equiv |-k/2\ra$ and $|e\ra \equiv |-k/2+1\ra$, the buffer state remains in the four-dimensional subspace spanned by $\{|g\ra,|e\ra\}^{\otimes 2}$, where the spin operators on each side reduce to $\oJ_z = (k/2)|g\ra\la g| + (k/2-1)|e\ra\la e|$ and $\oJ_+ = \oJ_-^\da = \sqrt{k}|e\ra\la g|$. The combined state of buffer and source pair after application of the weak caching operation reads as 
\begin{eqnarray}
\ket{\Psi}
    &=&\oU(\alpha,r\alpha) \frac{\ket{0}\ket{g}\ket{0}\ket{g} + \ket{1}\ket{g}\ket{1}\ket{g}}{\sqrt{2}} \nonumber\\
    &=& \frac{1}{\sqrt{2}}\left(\ket{0}\ket{g}\ket{0}\ket{g} + \left[\cos \left(\frac{\sqrt{k}\alpha}{2}\right)\E^{\I(1-r)k\alpha/2}\ket{1}\ket{g}-i\sin\left(\frac{\sqrt{k}\alpha}{2}\right)\E^{\I(1-r)\alpha/2}\ket{0}\ket{e}\right]^{\otimes 2}\right)\,,
\end{eqnarray}
up to the omitted global phase mentioned above. The reduced buffer state follows by tracing out the source pair,
\begin{eqnarray}
    \rho^{(k)}_1 
    &=& \frac{1+x^2}{2}\ket{gg}\bra{gg}+\frac{(1-x)^2}{2}\ket{ee}\bra{ee} - \frac{1-x}{2}\left[\E^{-\I(1-r)\alpha}\ketbra{gg}{ee} + h.c. \right] + \frac{x(1-x)}{2}\left[\ketbra{ge}{ge} + h.c.\right]\,,
\end{eqnarray}
with $x=\cos^2(\sqrt{k}\alpha/2) \in [0,1]$. The entanglement in terms of logarithmic negativity can be given exactly as
\begin{eqnarray}\label{apxeq:ent}
    E = \log_2\left[1+(x-1)^2\right] = \log_2 \left[1+\sin^4(\sqrt{k}\alpha/2)\right]\,,
\end{eqnarray}
which is independent of the phase rotation angle $\beta = r\alpha$. 
This implies that, for a buffer of sufficiently large size $k=(\pi/\alpha)^2$, the caching operation can extract all the entanglement from the source Bell pair. 
The same result applies if the buffer is initialized in $\ket{1}^{\otimes 2k}=\ket{e}$, as this merely swaps the roles of the basis states $\ket{0}\leftrightarrow\ket{1}$ and $\ket{g}\leftrightarrow\ket{e}$.

More generally, if the buffer is initialized in the product state $\ket{\phi}^{\otimes 2k}$ with $\ket{\phi}=\cos\theta\ket{0}+\E^{\I \delta}\sin\theta\ket{1}$, we can expand it on each side in terms of the Dicke states, 
\begin{eqnarray}
    \ket{\phi}^{\otimes k}=\sum_{m=-k/2}^{k/2} a_m \ket{m}\,,\qquad a_m=\sqrt{C_k^{k/2-m}}(\cos\theta)^{k/2-m}(\sin\theta)^{k/2+m} \E^{\I(k/2+m)\delta}\,,
\end{eqnarray}
where $C_N^K=N!/K!(N-K)!$ is the binomial coefficient. 
The caching unitary is block-diagonal with each block coupling $|m\ra|1\ra$ and $|m+1\ra|0\ra$ on the $A$-side and on the $B$-side, which facilitates an efficient numerical computation. 
For the simplest case of $k=1$, the combined source-buffer state after the caching operation is 
\begin{eqnarray}
\ket{\Psi}
    =&& \frac{1}{\sqrt{2}}\oU(\alpha,r\alpha) \left(\left[\ket{0}(\cos\theta\ket{g}+\E^{\I \delta}\sin\theta\ket{e})\right]^{\otimes 2} + \left[\ket{1}(\cos\theta\ket{g}+\E^{\I \delta}\sin\theta\ket{e})\right]^{\otimes 2}\right) \nonumber\\
    =&& \frac{1}{\sqrt{2}} \left[\E^{-\I(1-r)\alpha/4}\cos\theta\ket{0}\ket{g} 
    + \E^{\I(1-r)\alpha/4}\E^{\I \delta}\sin\theta \left(\cos \left(\frac{\alpha}{2}\right)\ket{0}\ket{e}-i\sin\left(\frac{\alpha}{2}\right)\ket{1}\ket{g} \right)\right]^{\otimes 2}\nonumber\\
    &+& \frac{1}{\sqrt{2}} \left[\E^{-\I(1-r)\alpha/4}\E^{\I \delta}\sin\theta\ket{1}\ket{e} 
    + \E^{\I(1-r)\alpha/4}\cos\theta\left(\cos\left(\frac{\alpha}{2}\right)\ket{1}\ket{g}-i\sin\left(\frac{\alpha}{2}\right)\ket{0}\ket{e}\right)\right]^{\otimes 2}\,,
\end{eqnarray}
which is valid for arbitrary $\alpha$. The reduced buffer state follows after tracing out the source pair or, alternatively, by directly applying the channel \eqref{eq:channelC_k1} to $\rho_0^{(2)} = |\phi\ra\la\phi|^{\otimes 2}$. The cached entanglement is shown in Fig.~\ref{fig:purestate} as a function of $(\theta,\delta)$ for various choices of $\alpha,\beta$, confirming our claim in the main text that the relative phase $\delta$ barely influences the caching performance. 
Figs.~\ref{fig:purestate}(a)-(d) correspond to iSWAP interactions ($\alpha=\beta$), complementing the results for $\delta=0$ shown in Fig.~\ref{fig:singlecopy}(b). 
Figs.~\ref{fig:purestate}(e)-(h) depict the behaviour for increasing $\beta$ starting from the SWAP case $\beta=0$, showing that the caching performance deteriorates with $\beta$, and is worst at $\beta=\alpha$ in (c).  

\begin{figure}
  \includegraphics[width=0.99\linewidth]{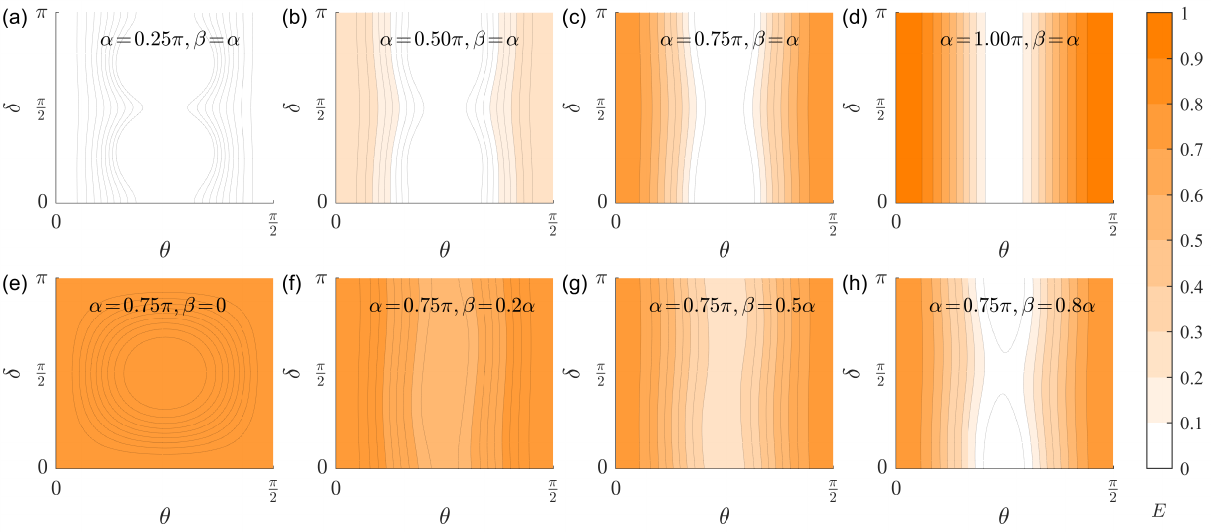}
  \caption{
  Entanglement cached from a single Bell pair in a 1-pair buffer for eight different caching unitaries (a)-(h) specified by the choice of the parameters $\alpha,\beta$. The two buffer qubits are initialized in the same pure state $\cos\theta\ket{0}+\E^{\I \delta}\sin\theta\ket{1}$, and we plot the cached entanglement in ebits of logarithmic negativity as a function of  $(\theta,\delta)$. 
  }
  \label{fig:purestate}
\end{figure}

The strong dependence of the cached entanglement on the superposition angle $\theta$ in the iSWAP case can be alleviated by considering larger buffers. Fig.~\ref{fig:singlecopy2} compares the caching performance of a 2-pair buffer at $\delta=0$ for (a) SWAP-type and (b) iSWAP-type operations. Compared to the 1-pair results in Fig.~\ref{fig:singlecopy}, the iSWAP case results in appreciable entanglement over a greater range of $\theta$-values.

\begin{figure}
  \includegraphics[width=0.49\linewidth]{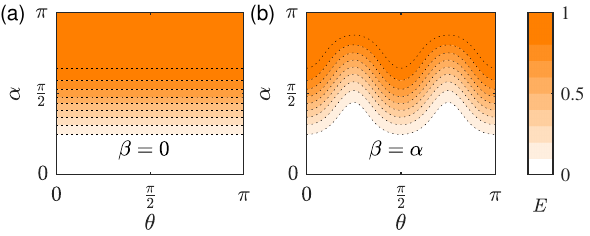}
  \caption{
  Entanglement cached from a single Bell pair in a 2-pair buffer by means of (a) a partial SWAP unitary and (b) a partial iSWAP of varying angles $\alpha$. The buffer qubits are initialized in a pure state $\cos\theta\ket{0}+\sin\theta\ket{1}$ of varying $\theta$. 
  }
  \label{fig:singlecopy2}
\end{figure}

Finally, we consider the repeated application of a weak caching unitary, $\oU^{n} (\alpha,r\alpha)$. In the limit of small $\alpha$, but finite $n\alpha$, we can commute the exponentials in \eqref{eq:U_weak} by virtue of the Baker-Campbell-Hausdorff formula and incur a correction of the order of $n\alpha^2$ that we assume small. Hence the combined source-buffer state becomes $|\Psi_n\ra = \E^{-\I n\alpha\oH}|\phi\ra^{\otimes 2k}|\psi_1\ra + \cO (n\alpha^2)$, with the effective Hamiltonian $\oH  = \oH_{xy}^{A}+\oH_{xy}^{B}+(1-r)(\oH_z^A + \oH_z^B)$ applied for an effective time $n\alpha$. If we break the interaction down into $k$ subsequent buffer pairs, the overall caching time scales like $nk\alpha$. The resulting entanglement is the same as for a single weak caching operation, Eq.~\eqref{apxeq:ent}, but now with $x=\cos^2(\sqrt{k}n\alpha/2)$ containing a finite (and possibly large) argument $n\alpha$ instead of $\alpha$. To cache 1 ebit, one thus requires $n=\pi/(\sqrt{k}\alpha)$ repetitions, which results in a caching time $nk\alpha/\pi=\sqrt{k}$, as plotted in Fig.~\ref{fig:onecopy_big}(b).

\section{Steady state of the multi-copy protocol} \label{Appx:multi}

Here we discuss the steady state of a $k$-pair buffer under repeated caching from a Bell pair sequence. Recall that the multi-copy caching protocol is a repeated application of the caching channel $\cC$, which defines the Markov chain 
\begin{eqnarray}
    \rho_n^{(2k)} = \cC \rho_{n-1}^{(2k)} = \cC^{\circ n} \rho_0^{(2k)}\,.
\end{eqnarray}
In Appendix~\ref{Appx:one_pair}, we have expressed the channel in terms of the Kraus operators $\oM_j = \la \psi_j|\oU(\alpha,\beta)|\psi_1\ra$ and given them explicitly for a 1-pair buffer. The Kraus operators also represent the Lindblad operators of the generator $\cL\rho := \cC[\rho] -\rho$ describing the discrete time increments of the buffer state.

After sufficiently many such increments, the buffer will generally converge towards a steady state $\rho_\infty^{(2k)}$ that obeys $\cL \rho_\infty^{(2k)} = 0$.
For $k=1$, we can analytically confirm that it is uniquely determined by $\oU(\alpha,\beta)$, as given in \eqref{eq:SS_1pair} in the main text. The corresponding logarithmic negativity \eqref{eq:E_SS_1pair} is also plotted in Fig.~\ref{fig:multicopy}(a) of the main text. The results for buffers of size $k=2,3$ shown in Fig.~\ref{fig:multicopy_big} were obtained by numerically solving for the kernel of the generator $\cL$. What we consistently observe is that, for iSWAP-type caching ($\beta=\alpha$), the buffer is always filled to its maximum $E=k\,$ebit, regardless of $\alpha>0$. In fact, the steady state is given by $k$ copies of the Bell state $|\psi_2\ra$, which we can prove for any $k$:

\begin{proposition}
    For iSWAP caching of entanglement from a sequence of Bell-$|\psi_1\ra$ pairs into a $k$-pair buffer, the $k$-ebit buffer state $\ket{\psi_2}^{\otimes k}$ is a steady state for every swap angle $\alpha$.
\end{proposition}
\begin{proof}
    Let us first factorize the partial iSWAP interaction ($\beta=\alpha$) into 
    \begin{eqnarray}
        \oU(\alpha,\alpha)
        =\prod_{j=1}^k \E^{-\I\alpha(\os_x^A \os_x^{A_j} + \os_y^A \os_y^{A_j} + \os_x^B \os_x^{B_j} + \os_y^B \os_y^{B_j})/4}
        =\prod_{j=1}^k \oR_{xy}^{AA_j}\otimes \prod_{j=1}^k \oR_{xy}^{BB_j} \equiv \oU^{(k)}\otimes \oU^{(k)}\,,
    \end{eqnarray}
    where $\oU^{(k)}$ acts on the $k+1$ qubits of buffer and source either on the $A$-side or on the $B$-side. Expanded in the basis of the source qubit, it can be written as
    \begin{eqnarray}\label{eq:Uk_iSWAP_expand}
        \oU^{(k)}:=\prod_{j=1}^k \oU_{j}
        =\ketbra{0}{0}\otimes\oU_{00}^{(k)}
        +\ketbra{1}{1}\otimes\oU_{11}^{(k)}
        +\ketbra{0}{1}\otimes\oU_{01}^{(k)}
        +\ketbra{1}{0}\otimes\oU_{10}^{(k)}\,.
    \end{eqnarray}
    The expansion "coefficients" $\oU_{ij}^{(k)}$ can be obtained iteratively, starting from $k=1$:
    \begin{eqnarray}
        \oP_{1} &\equiv& \oU_{00}^{(1)} = \la 0|\oR_{xy}(\alpha)|0\ra = \ketbra{0}{0}+\cos(\alpha/2)\ketbra{1}{1}\,, \nonumber \\
        \sigma'_{-} &\equiv& \oU_{10}^{(1)} = \la 1|\oR_{xy}(\alpha)|0\ra = -\I \sin(\alpha/2)\ketbra{0}{1}\,, \nonumber \\
        \oP_{2} &\equiv& \oU_{11}^{(1)} = \la 1|\oR_{xy}(\alpha)|1\ra = \ketbra{1}{1}+\cos(\alpha/2)\ketbra{0}{0}\,, \nonumber \\
        \sigma'_{+} &\equiv& \oU_{01}^{(1)} = \la 0|\oR_{xy}(\alpha)|1\ra = -\I \sin(\alpha/2)\ketbra{1}{0}\,.
    \end{eqnarray}
    Higher $k$-terms follow by virtue of the recursion relations
    \begin{eqnarray}
        \oU_{00}^{(k)}&=&\oU_{00}^{(k-1)}\otimes\oP_{1}+\oU_{10}^{(k-1)}\otimes\sigma'_{+}\,, \nonumber \\
        \oU_{10}^{(k)}&=&\oU_{00}^{(k-1)}\otimes\sigma'_{-}+\oU_{10}^{(k-1)}\otimes\oP_{2}\,, \nonumber \\
        \oU_{11}^{(k)}&=&\oU_{11}^{(k-1)}\otimes\oP_{2}+\oU_{01}^{(k-1)}\otimes\sigma'_{-}\,, \nonumber \\
        \oU_{01}^{(k)}&=&\oU_{11}^{(k-1)}\otimes\sigma'_{+}+\oU_{01}^{(k-1)}\otimes\oP_{1}\,.
    \end{eqnarray}
    We will now derive how the partial iSWAP acts on the buffer state $|\psi_2\ra^{\otimes k}$. Our claim is that
        \begin{eqnarray}
        &&\left[\oU_{00}^{(k)}\otimes\oU_{00}^{(k)}+\oU_{01}^{(k)}\otimes\oU_{01}^{(k)}\right]\ket{\psi_2}^{\otimes k}=\ket{\psi_2}^{\otimes k}\,,\nonumber \\
        &&\left[\oU_{11}^{(k)}\otimes\oU_{11}^{(k)}+\oU_{10}^{(k)}\otimes\oU_{10}^{(k)}\right]\ket{\psi_2}^{\otimes k}=\ket{\psi_2}^{\otimes k}\,,\nonumber \\
        &&\left[\oU_{00}^{(k)}\otimes\oU_{10}^{(k)}+\oU_{01}^{(k)}\otimes\oU_{11}^{(k)}\right]\ket{\psi_2}=0\,,\nonumber \\
        &&\left[\oU_{10}^{(k)}\otimes\oU_{00}^{(k)}+\oU_{11}^{(k)}\otimes\oU_{01}^{(k)}\right]\ket{\psi_2}=0\,. \label{eq:inductionClaim}
        \end{eqnarray}
These are the only four relevant combinations, because we assume that the source Bell pair is in the state $|\psi_1\ra$. The proof of this claim follows by induction. At $k=1$, we can quickly verify
    \begin{eqnarray}
        &&\left[\oU_{00}^{(1)}\otimes\oU_{00}^{(1)}+\oU_{01}^{(1)}\otimes\oU_{01}^{(1)}\right]\ket{\psi_2}=\left(\oP_{1}\otimes\oP_{1}+\sigma'_{+}\otimes\sigma'_{+}\right)\ket{\psi_2}=\ket{\psi_2}\,, \nonumber \\
        &&\left[\oU_{11}^{(1)}\otimes\oU_{11}^{(1)}+\oU_{10}^{(1)}\otimes\oU_{10}^{(1)}\right]\ket{\psi_2}=\left(\oP_{2}\otimes\oP_{2}+\sigma'_{-}\otimes\sigma'_{-}\right)\ket{\psi_2}=\ket{\psi_2}\,, \nonumber \\
        &&\left[\oU_{00}^{(1)}\otimes\oU_{10}^{(1)}+\oU_{01}^{(1)}\otimes\oU_{11}^{(1)}\right]\ket{\psi_2}=\left(\oP_{1}\otimes\sigma'_{-}+\sigma'_{+}\otimes\oP_{2}\right)\ket{\psi_2}=0\,, \nonumber \\
        &&\left[\oU_{10}^{(1)}\otimes\oU_{00}^{(1)}+\oU_{11}^{(1)}\otimes\oU_{01}^{(1)}\right]\ket{\psi_2}=\left(\sigma'_{-}\otimes\oP_{1}+\oP_{2}\otimes\sigma'_{+}\right)\ket{\psi_2}=0\,.
    \end{eqnarray}
    Now suppose our claim \eqref{eq:inductionClaim} holds for $k-1$. We can then verify by elementary calculation that each of the four identities also holds for $k$. For example, the first line in \eqref{eq:inductionClaim} follows by
    \begin{eqnarray}
        \left[\oU_{00}^{(k)}\otimes\oU_{00}^{(k)}+\oU_{01}^{(k)}\otimes\oU_{01}^{(k)}\right]\ket{\psi_2}^{\otimes k} =&&\left[\oU_{00}^{(k-1)}\otimes\oU_{00}^{(k-1)}+\oU_{01}^{(k-1)}\otimes\oU_{01}^{(k-1)}\right]\ket{\psi_2}^{\otimes k-1}\otimes\oP_{1}\otimes\oP_{1}\ket{\psi_2}\nonumber\\
        &+&\left[\oU_{11}^{(k-1)}\otimes\oU_{11}^{(k-1)}+\oU_{10}^{(k-1)}\otimes\oU_{10}^{(k-1)}\right]\ket{\psi_2}^{\otimes k-1}\otimes\sigma'_{+}\otimes\sigma'_{+}\ket{\psi_2}\nonumber\\
        &+&\left[\oU_{00}^{(k-1)}\otimes\oU_{10}^{(k-1)}+\oU_{01}^{(k-1)}\otimes\oU_{11}^{(k-1)}\right]\ket{\psi_2}^{\otimes k-1}\otimes\oP_{1}\otimes\sigma'_{+}\ket{\psi_2}\nonumber\\
        &+&\left[\oU_{10}^{(k-1)}\otimes\oU_{00}^{(k-1)}+\oU_{11}^{(k-1)}\otimes\oU_{01}^{(k-1)}\right]\ket{\psi_2}^{\otimes k-1}\otimes\sigma'_{+}\otimes\oP_{1}\ket{\psi_2}\,,\nonumber\\
        =&&\ket{\psi_2}^{\otimes k-1}\otimes\left(\oP_{1}\otimes\oP_{1}+\sigma'_{+}\otimes\sigma'_{+}\right)\ket{\psi_2}+0 = \ket{\psi_2}^{\otimes k}\,.
    \end{eqnarray}
    The other three identities are obtained in a similar manner.
    We can thus conclude that the combined source-buffer state $|\psi_1\ra |\psi_2\ra^{\otimes k}$ is a fixed point of the partial iSWAP unitary, since
    \begin{eqnarray}
        \oU(\alpha,\alpha) \ket{\psi_1}\ket{\psi_2}^{\otimes k} &=& \frac{1}{\sqrt{2}}\Bigl\{\ket{00}\left[\oU_{00}^{(k)}\otimes\oU_{00}^{(k)}+\oU_{01}^{(k)}\otimes\oU_{01}^{(k)}\right]\ket{\psi_2}^{\otimes k}   +\ket{11}\left[\oU_{11}^{(k)}\otimes\oU_{11}^{(k)}+\oU_{10}^{(k)}\otimes\oU_{10}^{(k)}\right]\ket{\psi_2}^{\otimes k}\nonumber\\ &&+\ket{01}\left[\oU_{00}^{(k)}\otimes\oU_{10}^{(k)}+\oU_{01}^{(k)}\otimes\oU_{11}^{(k)}\right]\ket{\psi_2}^{\otimes k} +\ket{10}\left[\oU_{00}^{(k)}\otimes\oU_{10}^{(k)}+\oU_{01}^{(k)}\otimes\oU_{11}^{(k)}\right]\ket{\psi_2}^{\otimes k}\Bigr\}\nonumber\\
        &=&\ket{\psi_1}\ket{\psi_2}^{\otimes k}\,.
    \end{eqnarray}
    This implies that $\ket{\psi_2}^{\otimes k}$ is a steady state of the buffer for any $\alpha$-value.
\end{proof}

\section{Robustness of multi-copy protocol}\label{Appx:adv}
Here we provide additional details and results on the impact of parameter fluctuations, white noise, and transmission loss in the multi-copy buffering protocol. 

\begin{figure}
  \includegraphics[height=0.20\linewidth]{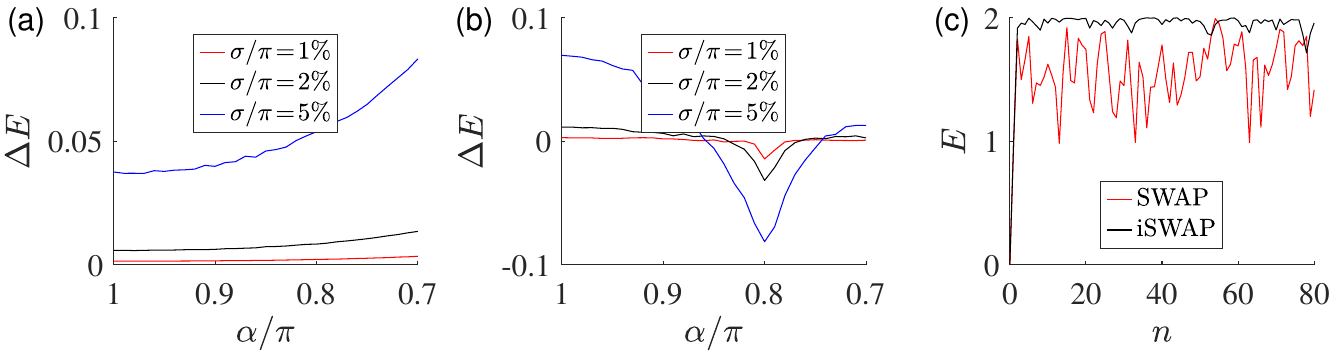}
  \caption{Loss of average steady-state entanglement due to parameter fluctuations of the swap angle $\alpha$ and relative phase $\beta$ for (a) partial SWAPs ($\Tilde{\alpha} \sim \mathcal{N}(\alpha,\sigma)$, $\Tilde{\beta} \sim \mathcal{N}(0,\sigma)$) and (b) partial iSWAPs ($\Tilde{\alpha} \sim \mathcal{N}(\alpha,\sigma)$, $\Tilde{\beta} \sim \mathcal{N}(\alpha,\sigma)$), compared to the zero-fluctuation case, $\Delta E = E(\alpha,\beta)-E(\Tilde{\alpha},\Tilde{\beta})$. 
  We average over 10000 random samples of 40-step trajectories. 
  (c) Random sample trajectories with the standard deviation $\sigma =5\% \times \pi$, for the partial SWAP and partial iSWAP at $\alpha/\pi=0.9$.
  }
  \label{fig:fluctuation2}
\end{figure}

The random fluctuations in the parameters $\alpha$ and $\beta$ are modeled as Gaussian distributions $\mathcal{N}(\mu,\sigma)$, where $\mu$ is the mean value and $\sigma$ is the standard variance.
In the main text, we discussed partial SWAPs ($\Tilde{\alpha} \sim \mathcal{N}(\alpha,\sigma)$, $\Tilde{\beta} \sim \mathcal{N}(0,\sigma)$) and partial iSWAPs ($\Tilde{\alpha} \sim \mathcal{N}(\alpha,\sigma)$, $\Tilde{\beta} \sim \mathcal{N}(\alpha,\sigma)$) under independent distributions with the same variance in each caching step $n$.
As shown in Fig.~\ref{fig:fluctuation} with $\sigma=5\% \pi$, there are two main results: 
(1) fluctuations do not lead to accumulating errors in entanglement caching; 
(2) under the same fluctuations, smaller values of $\alpha$ result in greater losses ($k-E$) in steady entanglement caching, but partial iSWAPs can tolerate smaller $\alpha$ values better than partial SWAPs. 

We make the following additional remarks. 
First, because arbitrary fixed partial iSWAPs can lead to maximal entanglement caching, it is clear that randomly fluctuating partial iSWAPs will incur entanglement loss, and stronger fluctuations with smaller $\alpha$ will lead to more losses; see Fig~\ref{fig:fluctuation2}(a). 
For partial SWAPs however, fluctuations can enhance entanglement caching for a certain region of $\alpha$-values, as shown in Fig~\ref{fig:fluctuation2}(b). 
Second, the random trajectories shown in Fig.~\ref{fig:fluctuation2}(c) illustrate that, in the presence of random fluctuations, the buffer no longer reaches a steady entangled state; only the trajectory-averaged cached entanglement does. 
Compared to partial SWAPs, the protocol with partial iSWAPs is more stable, showing the higher robustness of iSWAPs against the fluctuations.

\begin{figure}
  \includegraphics[width=0.95\linewidth]{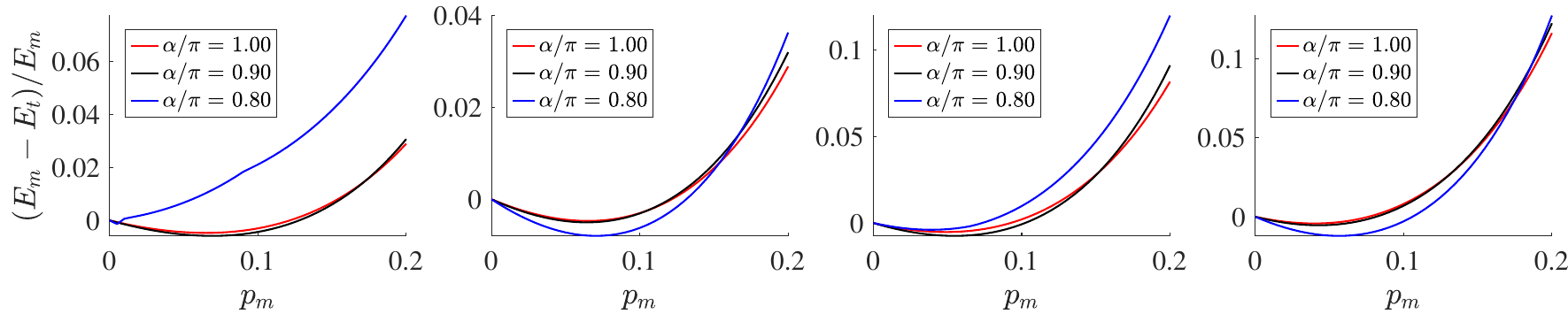}
  \caption{
  Comparison between the steady-state entanglement $E_t (p_t)$ subject to transmission noise of effective strength $p_t=p_m+(k-1)p_m^2$ and the entanglement $E_m (p_m)$ subject to storage noise of strength $p_m$, for various swap angles. We plot the cases of partial SWAPs and partial iSWAPs with a 2-pair buffer (a,b) and the same with a 3-pair buffer (c,d). 
  }
  \label{fig:noise2}
\end{figure}

For the noise during entanglement transmission and the noise during entanglement storage, we consider the identity noise model, which maps the source Bell state and the buffer state, respectively, to
\begin{eqnarray}
    \ket{\psi_1}\bra{\psi_1} &\longrightarrow& (1-p_t) \ket{\psi_1}\bra{\psi_1} + p_t \openone/4\,,\\
    \rho_n^{(2k)} &\longrightarrow& (1-p_m)\cC \rho_{n-1}^{(2k)} + p_m\openone/d\,.
\end{eqnarray}
Our choice is motivated by two reasons: (1) the identity noise (or white noise) physically corresponds to the global depolarizing channel, which is the common assumption especially in large-scale experiments \cite{urbanek_mitigating_2021, mi_information_2021}. 
(2) all noise can be transformed into identity noise form by applying additional noise \cite{foldager_can_2023, dalzell_random_2024}. 
An important observation is that both transmission and storage noise lead to similar behavior in that the buffer reaches a steady state with finite entanglement loss. 
Thus, for the sake of simplicity in possible applications, the storage noise can be included in the transmission noise by increasing the transmission noise parameter $p_t$ approximately as 
\begin{eqnarray}
    p_t \longrightarrow p_t + p_m + (k-1)p_m^2\,.
\end{eqnarray}
The precise difference in buffered steady-state entanglement between this approximation and the actual impact of storage noise at strength $p_m$ is shown in Fig.~\ref{fig:noise2}, assuming there is only storage noise. We compare various swap angles for partial SWAPs and partial iSWAPs with a 2-pair buffer (a,b) and with a 3-pair buffer (c,d).

\begin{figure}
  \includegraphics[width=0.70\linewidth]{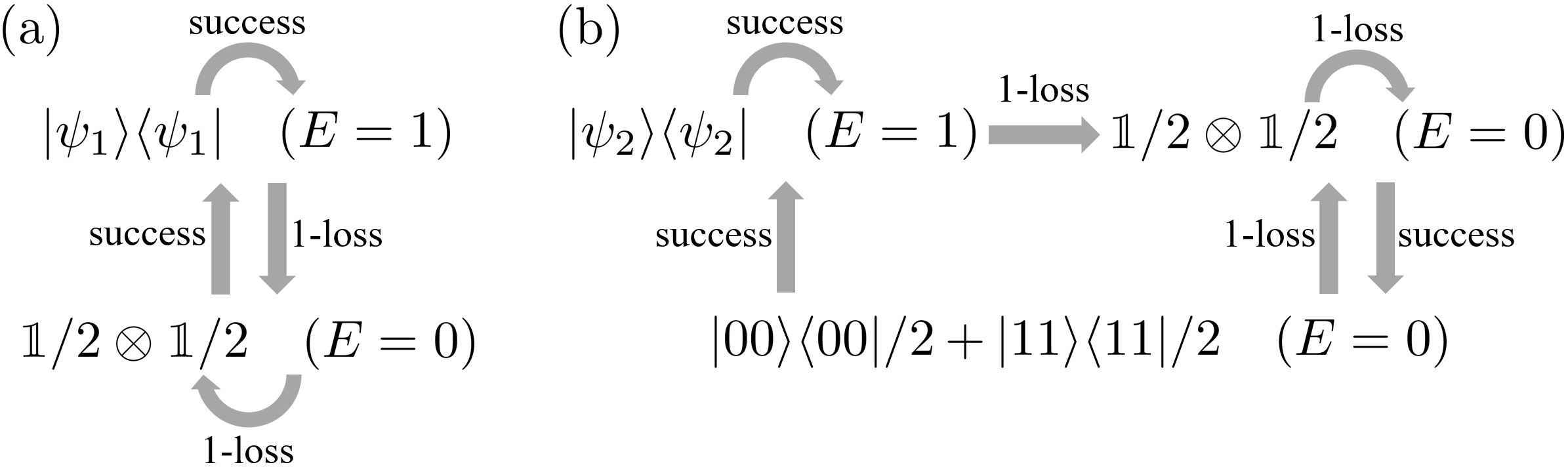}
  \caption{
  Diagram of state transformations between a filled and a depleted 1-pair buffer upon (a) full SWAP and (b) full iSWAP caching that can be successful ("success") or subject to single-sided loss events ("1-loss"). Two-sided loss events never change the buffer state.
  }
  \label{fig:statechange}
\end{figure}

\begin{figure}
  \includegraphics[width=0.6\linewidth]{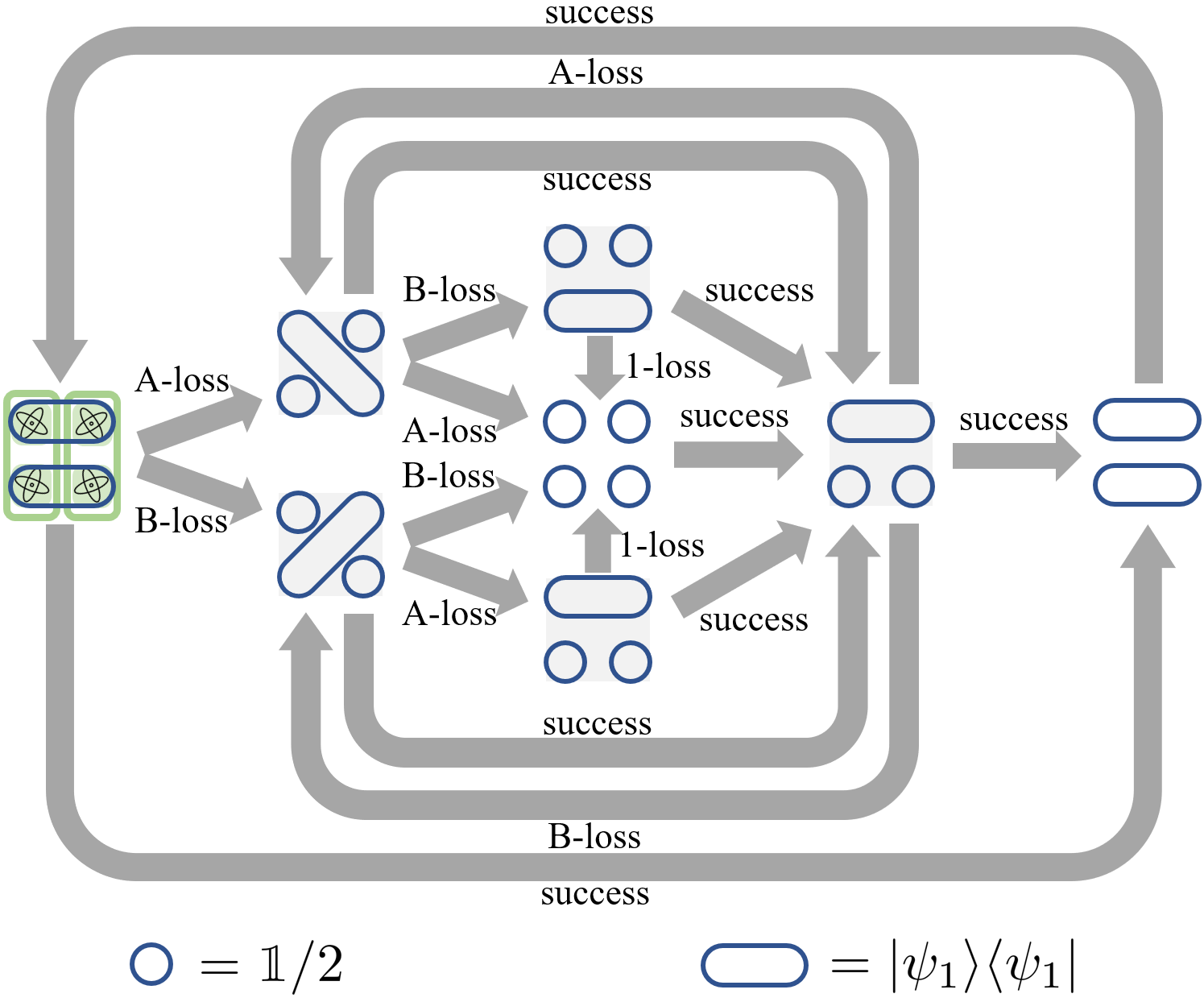}
  \caption{
  Diagram of state transformations between a filled and a depleted 2-pair buffer upon full SWAP caching that can be successful ("success") or subject to single-sided loss events ("1-loss", or "A-loss" and "B-loss" if the loss occurs on side A and B, respectively). Two-sided loss does not affect the buffer state. Grey-shaded nodes denote half-filled 1-ebit buffer states.
  }
  \label{fig:statechange2}
\end{figure}
Next, we consider transmission loss, which is typically more relevant than transmission noise in a quantum network scenario. 
First, we compare the 1-ebit success rate $q_n$ between SWAP and iSWAP caching. In the main text, we found the simple expression \eqref{eq:qn_SWAP} for the success rate using full SWAP caching of a 1-pair buffer, which exceeds the single-shot success probability $p^2$. At the same time, we noticed that iSWAPs are more sensitive to transmission loss with $q_n < p^2$ consistently. The reason lies in the fact that, once the buffer assumes a maximally mixed state (after a single-sided loss event), the full iSWAP caching protocol requires two successful caching operations in a row to refill the buffer, while the full SWAP protocol requires only one. See the diagram in Fig.~\ref{fig:statechange} for a visual representation of both cases. 

Concretely, a full iSWAP protocol with single-sided loss transforms between the buffer states  $\ketbra{\psi_2}{\psi_2}$ (filled), $\openone/2\otimes\openone/2$ (depleted), and $\ketbra{00}{00}/2+\ketbra{11}{11}/2$ (depleted). Consequently, the 1-ebit success rate after $n$ steps requires two two-sided successes with an arbitrary number of up to $n-2$ two-sided losses in between or thereafter,
\begin{equation}
    q'_n = \sum_{k=0}^{n-2} \left[\sum_{\ell=0}^{n-2-k} (1-p)^{2\ell} \right] p^2 (1-p)^{2k} p^2 = \frac{1-(1-p)^{2(n-1)}-(n-1)p(2-p)(1-p)^{2(n-1)}}{(2-p)^2}p^2 \xrightarrow{n\gg 1} \frac{p^2}{(2-p)^2},
\end{equation}
This is smaller than the single-shot success probability $p^2$. 

For larger buffers, we generally obtain a greater robustness against single-sided losses and thus a greater 1-ebit success rate. The diagram in Fig.~\ref{fig:statechange2} illustrates the possible state transformations from a filled to a depleted 2-pair buffer upon full SWAP caching under loss. In the iSWAP case, the diagram looks the same, and the only difference lies in the 1-ebit states (grey-shaded nodes): For SWAPs they are permutations of $|\psi_1\ra\la\psi_1|\otimes \openone/4$, whereas for iSWAPs they are permutations of  $\ket{\psi_1}\bra{\psi_1}\otimes(\ketbra{00}{00}+\ketbra{11}{11})/4+\ket{\psi_2}\bra{\psi_2}\otimes(\ketbra{01}{01}+\ketbra{10}{10})/4$. 
Both cases are therefore equivalent in terms of success rate, and we compute its asymptotic value as $q_{n\gg 1} \approx 63\,\%$.

Numerically, we found a similar behaviour for buffer size $k=4$ (but not $3$). Both full SWAP and full iSWAP caching can reach a 1-ebit success rate of about $89\,\%$, suggesting that the equivalence of both protocols might occur consistently at even $k$. Figure \ref{fig:loss} (c) shows $q_{10} (E\lesssim 1)$ for partial SWAP and iSWAP caching of varying angles $\alpha$ at $k=4$; panels (a) and (b) are copied from Fig.~\ref{fig:probability} for easy comparison. In (d), we show the convergence of the success rate to its asymptotic value with the number of steps $n$, comparing full SWAPs (solid) and full iSWAPs (dotted) and buffers up to $k=4$. Both cases differ only for odd $k$. Moreover, our choice of $n=10$ in (a)-(c) approximates the asymptotic regime sufficiently well.

Moreover, Fig.~\ref{fig:highloss} shows the success rates for the case of more severe losses, $p=0.1$, closer to what is currently achieved in state-of-the-art quantum memory implementations with photonic transmission channels. The relative improvement of the asymptotic 1-ebit success rate compared to the single-shot probability $p^2$ is more pronounced and also increases with buffer size, but reaching it also requires more caching steps. We choose $n=66$, letting $(1-p)^{2n}\leq 10^{-6}$, to converge its asymptotic value.

\begin{figure}[tb]
  \includegraphics[width=0.99\linewidth]{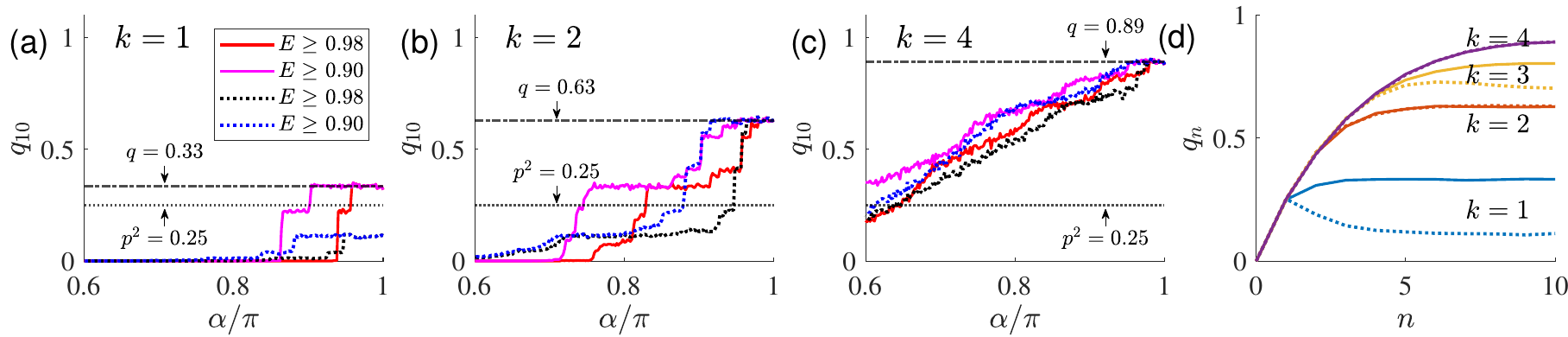}
  \caption{
  Probability to cache $E$ ebits of entanglement by partial SWAPs (solid lines) and iSWAPs (dash lines) of varying swap angle $\alpha$ with a sequence of Bell pairs. The horizontal dotted line marks the probability to successfully transmit the Bell pair, the dash-dotted line gives the asymptotic probability to cache 1 ebit in an arbitrarily long sequence at $\alpha=\pi$. We consider a single-side transmission probability $p=0.5$ for (a) a 1-pair buffer (b) a 2-pair buffer (c) a 4-pair buffer and (d) the convergence of the asymptotic probability with the number of steps $n$. Panels (a),(b) are taken from Fig.~\ref{fig:probability} for direct comparison.}
  \label{fig:loss}
  \vspace{15pt}
  \includegraphics[width=0.99\linewidth]{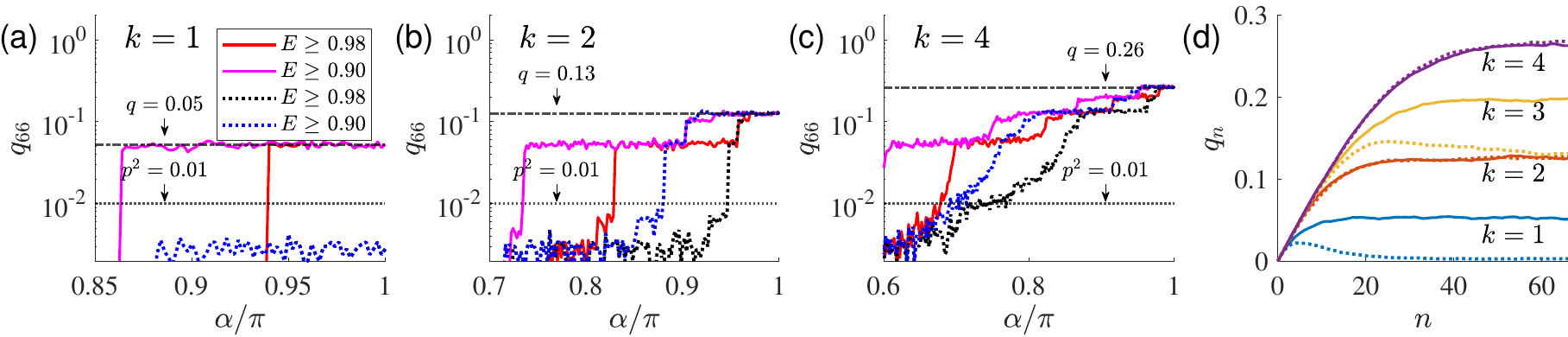}
  \caption{
  Probability to cache $E$ ebits of entanglement by partial SWAPs (solid lines) and iSWAPs (dash lines) of varying swap angle $\alpha$ with a sequence of Bell pairs. Compared to Fig.~\ref{fig:loss}, we here consider a smaller single-side transmission probability $p=0.1$, closer to the state of the art.}
  \label{fig:highloss}
\end{figure}

Finally, we extend our analysis to include the combined effects of transmission loss and memory noise. We set the memory-noise parameter to a fixed, moderate value of $p_m=0.01$ and the single-side transmission success probability to $p=0.5$. In Fig.~\ref{fig:mixture}, we present numerical results highlighting several key points.
First, we observe that small levels of memory noise ($p_m=0.01$) do not substantially alter the overall behavior previously observed in the transmission-loss-only scenario. The protocol remains robust, reliably reaching a stable steady-state entanglement.
Second, as expected, memory noise slightly reduces the achievable steady-state entanglement, thereby lowering the success probability for reaching very high entanglement thresholds (e.g., $E \geq 0.98$ ebits). However, for somewhat lower entanglement thresholds (e.g., $E \geq 0.9$ ebits), the success probability remains essentially unaffected.
Additionally, we consider the combined scenario of transmission loss and an additional small transmission noise ($p_t=0.01$). This inclusion does not significantly alter the observed behavior, further demonstrating the robustness of the protocol under realistic quantum network conditions.

\begin{figure}[tb]
  \includegraphics[width=0.99\linewidth]{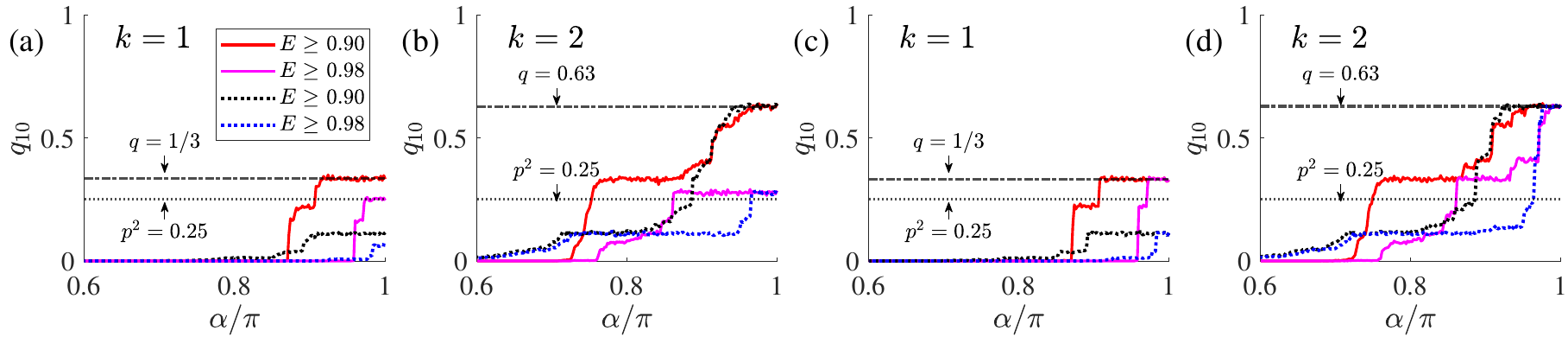}
  \caption{Probability to cache $E$ ebits of entanglement under combined imperfections with single-side transmission probability $p=0.5$. (a),(b) show combined transmission loss and memory noise ($p_m=0.01$), and (c),(d) show combined transmission loss and small transmission noise ($p_t=0.01$), each for buffer sizes $k=1$ and $k=2$, respectively. Partial SWAP (solid lines) and partial iSWAP (dashed lines) caching protocols are compared across varying swap angles $\alpha$. Compared to Fig.~\ref{fig:loss}, memory noise slightly reduces the success probability for caching very high entanglement thresholds ($E=0.98$ ebits) but does not significantly affect moderate thresholds ($E=0.9$ ebits). Likewise, the inclusion of small transmission noise does not substantially alter the overall performance.
  }
  \label{fig:mixture}
\end{figure}

\end{document}